\documentclass[a4paper,11pt,twoside]{article}
\usepackage[INRIA]{lipinria}
\usepackage{a4wide}
\usepackage{amsfonts}
\usepackage{amsmath}
\usepackage{amssymb}
\usepackage{amstext}
\usepackage{latexsym}
\usepackage{paralist}
\usepackage{url}
\usepackage{xspace}
\usepackage{algorithm2e}
\usepackage{subfigure}
\usepackage{setspace}
\usepackage{amsthm}

\usepackage{figlatex}

\newtheorem{theorem}{Theorem}

\newtheorem{definition}{Definition}

\AtBeginDocument{}

\def\ie{i.e.,\xspace}

\newcommand{\calc}{w}    
\newcommand{\comm}{\mathsf{\delta}} 
\newcommand{\speed}{s}   
\newcommand{\bwprocs}{Bs}      
\newcommand{\bwprocp}{Bp}      
\newcommand{\bwlinks}{bs}      
\newcommand{\bwlinkp}{bp}      
\newcommand{\dbrate}{rate} 
\newcommand{\dbsize}{\comm} 
\newcommand{\dbfreq}{f}    
\newcommand{\obj}{o}

\newcommand{\NN}{\mathcal{N}} 
\newcommand{\PP}{\mathcal{P}} 
\newcommand{\OO}{\mathcal{O}} 
\newcommand{\CC}{\mathcal{C}} 
\renewcommand{\SS}{\mathcal{S}} 
\newcommand{\RR}{\mathcal{R}} 
\newcommand{\Leaf}{\mathit{Leaf}}
\newcommand{\Child}{\mathit{Child}}
\newcommand{\Parent}{\mathit{Parent}}

\newcommand{\allocn}{a} 
\newcommand{\allocni}{\bar{a}} 
\newcommand{\download}{download}

\newcommand{\CONS}{\textsc{Constr}\xspace}
\newcommand{\CONSHOM}{\textsc{Constr-Hom}\xspace}
\newcommand{\CONSLAN}{\textsc{Constr-LAN}\xspace}
\newcommand{\NCONSHET}{\textsc{Non-Constr-Het}\xspace}
\newcommand{\NCONSHOM}{\textsc{Non-Constr-Hom}\xspace}
\newcommand{\NCONSLAN}{\textsc{Non-Constr-LAN}\xspace}
\newcommand{\NCONS}{\textsc{Non-Constr}\xspace}
\newcommand{\LDT}{\textsc{LDTree}\xspace}
\newcommand{\HOMA}{\textsc{HomA}\xspace}
\newcommand{\HOMS}{\textsc{HomS}\xspace}
\newcommand{\NOCOMA}{\textsc{NoComA}\xspace}
\newcommand{\CLDTHom}{\textsc{C-LDT-Hom}\xspace}
\newcommand{\CLDTHomDec}{\textsc{C-LDT-Hom-Dec}\xspace}

\usepackage{setspace}

\newif\ifremark
\long\def\remark#1{
\ifremark%
        \begingroup%
        \dimen0=\columnwidth
        \advance\dimen0 by -1in%
        \setbox0=\hbox{\parbox[b]{\dimen0}{\protect\em #1}}
        \dimen1=\ht0\advance\dimen1 by 2pt%
        \dimen2=\dp0\advance\dimen2 by 2pt%
        \vskip 0.25pt%
        \hbox to \columnwidth{%
                \vrule height\dimen1 width 3pt depth\dimen2%
                \hss\copy0\hss%
                \vrule height\dimen1 width 3pt depth\dimen2%
        }%
        \endgroup%
\fi}

\remarktrue

\begin{document}
\RRNo{6578}

\RRInumber{2008-20}

\RRItitle{Resource Allocation Strategies for \\  In-Network
  Stream Processing}

\RRItitre{Stratégies d'allocation de resources pour le traitement de flux en r\'eseaux}

\RRIthead{In-Network Stream Processing}
\RRIahead{A. Benoit \and H. Casanova \and V. Rehn-Sonigo\and Y. Robert}

\RRIauthor{Anne Benoit \and Henri Casanova \and Veronika Rehn-Sonigo \and Yves Robert}

\RRIdate{July 2008}

\RRIkeywords{in-network stream processing, trees of operators,
operator mapping, optimization, complexity results, polynomial
heuristics.}

\RRImotscles{traitement de flux en réseau, arbres d'opérateurs,
placement d'opérateurs, optimisation, résultats de complexité,
heuristiques polynomiales.}


\RRIabstract{In this paper we consider the operator mapping problem for in-network
stream processing applications. In-network stream processing consists in
applying a tree of operators in steady-state to multiple data objects that
are continually updated at various locations on a network. Examples of
in-network stream processing include the processing of data in a sensor
network, or of continuous queries on distributed relational databases.
We study the operator mapping problem in a ``constructive'' scenario,
i.e., a scenario in which one builds a platform dedicated to the application
buy purchasing processing servers with various costs and capabilities. The
objective is to minimize the cost of the platform while ensuring that the
application achieves a minimum steady-state throughput.

The first contribution of this paper is the formalization of a set of
relevant operator-placement problems as linear programs, and a proof that
even simple versions of the problem are NP-complete.  Our second
contribution is the design of several polynomial time heuristics, which are
evaluated via extensive simulations and compared to theoretical bounds for
optimal solutions. }

\RRIresume{Dans ce travail nous nous intéressons au problème de
  placement des applications de traitement de flux en
  réseau. Ce probl\`eme consiste \`a appliquer en régime permanent un arbre d'opérateurs \`a des
   données multiples qui sont mise \`a jour en permanence dans les
  différents emplacements du réseau. Le traitement de données
  dans les réseaux de détecteurs ou le traitement de requêtes dans les
  bases de données relationnelles sont des exemples d'application.
  Nous étudions le placement des
  opérateurs dans un scénario ``constructif'', i.e., un scénario dans
  lequel la plate-forme pour l'application est construite au fur et \`a
  mesure en achetant des serveurs de calcul ayant un vaste choix de
  coûts et de capacités. L'objectif est la minimisation du coût de la
  plate-forme en garantissant que l'application atteint un débit
  minimal en régime permanent.

La première contribution de cet article est la formalisation d'un
ensemble pertinent de problèmes opérateur-placement sous forme d'un
programme linéaire ainsi qu'une preuve que même les instances simples
du problème sont NP-compl\`etes. La deuxième contribution est la
conception de plusieures heuristiques polynomiales qui sont évaluées a
l'aide de simulations extensives et comparées aux bornes théoriques
pour des solutions optimales.  }

\RRItheme{\THNum}
\RRIprojet{GRAAL}

\RRImaketitle

\section{Introduction}
\label{sec.intro}

In this paper we consider the execution of applications structured as
trees of operators. The leaves of the tree correspond to basic data objects that are
spread over different servers in a
distributed network.  Each internal node in the tree denotes the
aggregation and combination of the data from its children, which in turn
generate new data that is used by the node's parent. The computation is
complete when all operators have been applied up to the root node, thereby
producing a final result.  We consider the scenario in which the basic data
objects are constantly being updated, meaning that the tree of operators
must be applied continuously. The goal is to produce final results at some
desired rate.

The above problem, which is called \emph{stream
processing}~\cite{badcock_VLDB_2004}, arises in several domains.  An important
domain of application is the acquisition and refinement of data from a set
of sensors~\cite{srivastava_PODS2005, madden_ICMD_2003, bonnet_CMDB_2001}.
For instance, \cite{srivastava_PODS2005} outlines a video surveillance
application in which the sensors are cameras located in different locations
over a geographical area. The goal of the application could be to show an
operator monitored area in which there is significant motion between
frames, particular lighting conditions, and correlations between the
monitored areas.  This can be achieved by applying several operators
(filters, image processing algorithms) to the raw images, which are
produced/updated periodically.  Another example arises in the area of
network monitoring~\cite{cranor_ICMD_2002,vanRennesse_IPTPS_2002,cooke_USENIX_2006}.
In this case the sources of data are routers that produce streams of data
pertaining to packets forwarded by the routers.  One can often view stream
processing as the execution of one of more ``continuous queries'' in the
relational database sense of the term (e.g., a tree of join and select
operators). A continuous query is applied continuously, i.e., at a
reasonably fast rate, and returns results based on recent data generated by
the data streams. Many authors have studies the execution of continuous
queries on data streams
~\cite{Babu_SIGMODRECORD_2001,liu_tkde1999,chen02design,Plale_TPDD_2003,kramer_COMAD05}.

In practice, the execution of the operators on the data streams must be distributed over the
network.  In some cases, for instance in the aforementioned video
surveillance application, the cameras that produce the basic objects do not
have the computational capability to apply any operator effectively. Even
if the servers responsible for the basic objects have sufficient
capabilities, these objects must be combined across devices, thus requiring
network communication.  A simple solution is to send all basic objects to a
central compute server, but it proves unscalable for many applications due
to network bottlenecks. Also, this central server may not be able to meet
the desired target rate for producing results due to the sheer amount of
computation involved.  The alternative is then to distribute the execution
by mapping each node in the operator tree to one or more compute servers in
the network (which may be distinct or co-located with the devices that
produce/store and update the basic objects).  One then talks of
\emph{in-network} stream processing. Several in-network stream processing
systems have been
developed~\cite{abadi2005design,medusa,pier,gates,nath-irisnet,vanRennesse_IPTPS_2002,NiagaraCQ,MORTAR}.
These systems all face the same question: to which servers should one map
which operators?

In this paper we address the operator-mapping problem for in-network stream
processing.  This problem was studied in~\cite{srivastava_PODS2005,
Pietzuch_ICDE06,ahmad_VLDB_2004}.  The work in~\cite{Pietzuch_ICDE06}
studied the problem for an ad-hoc objective function that trades off
application delay and network bandwidth consumption. In this paper we study
a more general objective function. We first enforce the constraint that the
rate at which final results are produced, or \emph{throughput}, is above a
given threshold.  This corresponds to a Quality of Service (QoS)
requirement of the application, which is almost always desirable in
practice (e.g., up-to-date results of continuous queries must be available at
a given frequency).  Our objective is to meet this constraint while
minimizing the ``overall cost'', that is the amount of resources used to
achieve the throughput. For instance, the cost could be simply the total
number of compute servers, in the case when all servers are identical and
network bandwidth is assumed to be free.

We study several variations of the operator-mapping problem.  Note that in
all cases basic objects may be replicated at multiple locations, i.e.,
available and updated at these locations.  In terms of the computing platform
one can consider two main scenarios. In the first scenario, which we term
``constructive'', the user can build the platform from scratch using
off-the-shelf components, with the goal of minimizing monetary cost while
ensuring that the desired throughput is achieved. In the second scenario,
which we term ``non-constructive'', the platform already exists and the
goal is to use the smallest amount of resources in this platform while
achieving the desired throughput. In this case we consider platforms that
are either fully homogeneous, or with a homogeneous network but
heterogeneous compute servers, or fully heterogeneous.
  In terms of the tree of operators, we
consider general binary trees and discuss relevant special cases (e.g.,
left-deep trees~\cite{ioannidis96query,chaudhuri98overview,Deshpande2007}).

Our main contributions are the following:
\begin{itemize}
 \item we formalize a set of relevant operator-placement problems;
 \item we establish complexity results (all problems turn out to be
   NP-complete);
\item we derive an integer linear programming formulation of the problem;
 \item 
we propose several heuristics for the constructive scenario; and
\item we compare heuristics through extended simulations, and assess
  their absolute performance with respect to the optimal solution
  returned by the linear program.
\end{itemize}

In Section~\ref{sec.models} we outline our application and platform models
for in-network stream processing. Section~\ref{sec.probs} defines several
relevant resource allocation problems, which are shown to be NP-complete in
Section~\ref{sec.np}. Section~\ref{sec.ilp} derives an integer linear
programming
formulation of the resource allocation problems. We present several
heuristics for solving one of our resource allocation problems in
Section~\ref{sec.heuristics}. These heuristics are evaluated in
Section~\ref{sec.results}.
Finally, we conclude the paper in
Section~\ref{sec.conclusion} with a brief summary of our results and future
directions for research.

\section{Models}
\label{sec.models}

\subsection{Application model}
\label{sec.model.appli}

\begin{figure}
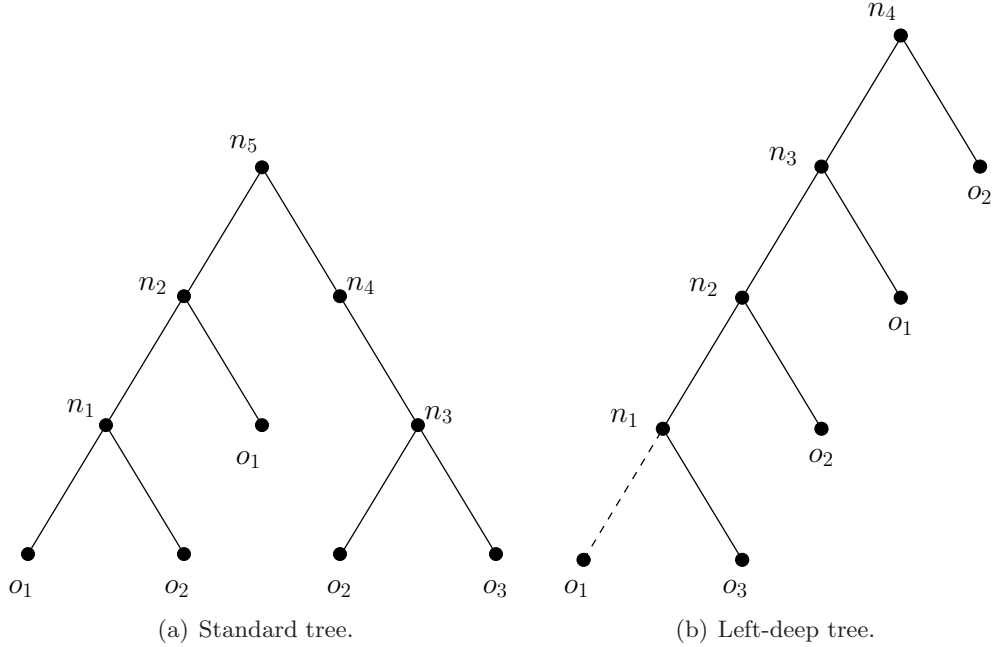

   \centering
    \vspace{-0.5cm}
   \subfigure[Standard tree.]{
     \includegraphics[width=0.42\textwidth]{sample-tree.fig}
     \label{fig.tree1}
   }$\quad$
   \subfigure[Left-deep tree.]{
     \includegraphics[width=0.36\textwidth]{leftdeeptree.fig}
     \label{fig.leftdeeptree}
   }
\vspace{-0.2cm}
   \caption{Examples of applications structured as a binary tree of operators.}
\label{fig.binarytree}
\end{figure}


We consider an application that can be represented as a set of
operators,~$\NN$.
These operators are organized as a binary tree, as
shown in Figure~\ref{fig.binarytree}.
Operations are initially performed on basic objects, which are
made available and continuously updated at given locations in a
distributed network.
We denote the set of basic objects $\OO=\{\obj_1, \obj_2, \obj_3, \dots\}$.
The leaves of the tree are thus the basic objects, and several leaves may
correspond to the same object, as illustrated in the figure.
Internal nodes (labeled $n_1, n_2, n_3, \dots$) represent operator
computations.   We call those operators that have at least one
basic object as a child in the tree an \emph{al-operator} (for ``almost leaf'').
For an operator $n_i$ we define:
\begin{compactitem}
\item $\Leaf(i)$: the index set of the basic objects that are needed
for the computation of $n_i$, if any;
\item $\Child(i)$: the index set of the node's children in $\NN$, if any;
\item $\Parent(i)$: the index of the node's parent in $\NN$, if it exists.
\end{compactitem}
We have the constraint that $|\Leaf(i)| + |\Child(i)| \leq 2$ since our tree is binary.
All functions above are extended to sets of nodes:
$f(I) = \cup_{i\in I} f(i)$, where $I$ is an index set and $f$ is
$\Leaf$, $\Child$ or $\Parent$.

The application must be executed so that it produces final results,
where each result is generated by executing the whole operator tree
once, at a target rate. We call
this rate the application \emph{throughput} $\rho$ and the specification of the
target throughput is a QoS requirement for the application.
Each operator $n_i \in \NN$ must compute (intermediate) results at a rate
at least as
high as the target application throughput.  Conceptually, a server executing
an operator consists of two
concurrent threads that run in steady-state:
 \begin{itemize}
   \item One thread periodically downloads the most recent copies of the basic
   objects corresponding to the operator's leaf children, if any.  For our example
   tree in Figure~\ref{fig.tree1}, $n_1$ needs to download $\obj_1$ and
   $\obj_2$ while $n_2$ downloads only $\obj_1$ and $n_5$ does not
download any basic
   object. Note that these downloads may simply amount to constant
   streaming of data from sources that generate data streams. Each
   download has a prescribed cost in terms of bandwidth based
   on application QoS requirements (e.g., so that computations are
   performed using sufficiently up-to-date data).  A basic object $\obj_k$ has a size
   $\dbsize_k$ (in bytes) and needs to be downloaded by the processors that use it with frequency $\dbfreq_k$. Therefore,
   these basic object downloads consume an amount of bandwidth equal to $\dbrate_k =
   \dbsize_k \times \dbfreq_k$ on each network link and network card
   through which this object is communicated.  

   \item Another thread receives data from the operator's non-leaf children, if any, and
    performs some computation using downloaded basic objects and/or
    data received from other operators. 
    The operator produces some output
   that needs to be passed to its parent operator. The computation of operator
   $n_i$ (to evaluate the operator once) requires $\calc_i$ operations, and
   produces an output of size $\comm_{i}$.

\end{itemize}


In this paper we sometimes consider \emph{left-deep} trees, \ie binary
trees in which the right child of an operator is always a leaf. These
trees arise in practical settings~\cite{ioannidis96query,chaudhuri98overview,Deshpande2007} and we show an
example of left-deep tree in Figure~\ref{fig.leftdeeptree}.
Here $\Child(i)$ and $\Leaf(i)$ have cardinal $1$ for every operator
$n_i$ but for the bottom-most operator, $n_j$, for which
$\Child(j)$ has cardinal $0$, and $\Leaf(j)$ has cardinal $1$ or $2$
depending on the application.


\subsection{Platform model}
\label{sec.model.platform}

The target distributed network is a fully connected graph (\ie a clique)
interconnecting a set of resources $\RR = \PP \cup \SS$, where $\PP$
denotes compute servers, or \emph{processors} for short, and $\SS$ denotes
data servers, or \emph{servers} for short. Servers hold and update basic
objects, while processors apply operators of the application tree.  Each
server $S_l \in \SS$ (resp. processor $P_u\in \PP$)  is interconnected to the
network via a network card with maximum bandwidth $\bwprocs_l$
(resp. $\bwprocp_u$).   The network link from a server $S_l$ to a
processor $P_u$ has bandwidth $\bwlinks_{l,u}$; on such links the
server sends data and the processor receives it. The link between
two distinct processors $P_u$ and $P_v$ is
bidirectional and it has bandwidth $\bwlinkp_{u,v}(=\bwlinkp_{v,u})$ shared
by communications in both directions.
In addition, each processor $P_u \in \PP$ is
characterized by a compute speed $\speed_u$.

Resources operate under the full-overlap, bounded multi-port
model~\cite{HongPrasanna07}. In this model, a resource $R_u$ can be involved in
computing, sending data, and receiving data simultaneously. Note that
servers only send data, while processors engage in all three activities.
A resource $R$, which is either a server or a processor, can be connected to multiple network links (since we assume a clique
network). The ``multi-port'' assumption states that $R$ can send/receive
data simultaneously on multiple network links. The ``bounded'' assumption
states that the total transfer rate of data sent/received by resource $R$
is bounded by its network card bandwidth ($\bwprocs_l$ for server $S_l$, or $\bwprocp_u$ for processor $P_u$).

\subsection{Mapping Model and Constraints}

Our objective is to map operators, i.e., internal nodes of the application
tree, onto processors. As explained in Section~\ref{sec.model.appli}, if a
tree node has leaf children it must continuously download up-to-date basic
objects, which consumes bandwidth on its processor's network card.  Each
used processor is in charge of one or several operators.  If there is only
one operator on processor $P_u$, while the processor computes for the
$t$-th final result it sends to its parent (if any) the data corresponding to
intermediate results for the $(t-1)$-th final result. It also receives data
from its non-leaf children (if any) for computing the $(t+1)$-th final
result. All three activities are concurrent (see
Section~\ref{sec.model.platform}).  Note however that different operators can
be assigned to the same processor. In this case, the same overlap happens,
but possibly on different result instances (an operator may be applied for
computing the $t_1$-th result while another is being applied for computing
the $t_2$-th). The time required by each activity must be summed for all
operators to determine the processor's computation time.


We assume that a basic object can be duplicated, and thus be available
and updated at multiple servers. We assume that duplication of basic objects is
achieved in some out-of-band manner specific to the target application. For instance,
this could be achieved
via the use of a distributed database infrastructure that enforces
consistent data replication.  In this case, a processor can choose among
multiple data sources when downloading a basic object.  Conversely, if two
operators have the same basic object as a leaf child and are mapped to
different processors, they must both continuously download that object (and
incur the corresponding network overheads).

We denote the mapping of the operators in $\NN$ onto the processors in
$\PP$ using an allocation function $\allocn$: $\allocn(i) = u$ if operator
$n_i$ is assigned to processor $P_u$. Conversely, $\allocni(u)$ is the
index set of operators mapped on $P_u$: $\allocni(u) = \{i|\allocn(i)=u\}$.

We also introduce new
notations to describe the location of basic objects. Processor~$P_u$
may need to download some basic objects from some servers. We use
$\download(u)$ to denote the set of $(k,l)$ couples where object
$\obj_k$ is downloaded by processor $P_u$ from server~$S_l$.

Given these notations we can now express the constraints for the
required application throughput, $\rho$. Essentially, each processor
has to communicate and compute fast enough to
achieve this throughput, which is expressed via a set of constraints.
Note that a communication occurs only when a child or the parent of a
given tree node and this node are mapped on different processors. In other
terms, we neglect intra-processor communications.

\begin{itemize}

\item Each processor $P_u$ cannot exceed its computation capability:
\begin{equation}
\label{const-comp}
\forall P_u \in \PP,\quad \sum_{i \in \allocni(u)}
\rho \cdot  \frac{\calc_i}{\speed_u} \leq 1
\end{equation}

  \item $P_u$ must have enough bandwidth capacity to perform all its basic
  object downloads and all communication with other processors. This
is expressed by the following constraint, in which the first term corresponds
to basic object downloads,
  the second term corresponds to inter-node communications when a tree
  node
is assigned to $P_u$ and its parent node is assigned to another
processor, and the third term corresponds
  to inter-node communications when a node is assigned to $P_u$ and
some of its children  nodes are assigned to another processor:\\
$\forall P_u \in \PP,\quad$
  \begin{equation}
\label{const-bpu}
  \sum_{(k,l) \in \download(u)} \dbrate_k +
     \sum_{j\in\Child(\allocni(u))\setminus \allocni(u)} \rho.\delta_j  +
     \sum_{j\in \Parent(\allocni(u))\setminus \allocni(u)~}\sum_{i\in
\Child(j)\cap \allocni(u)} \rho.\delta_i
    \leq \bwprocp_u
\end{equation}

  \item Server $S_l$ must have enough bandwidth capacity to
  support all the downloads of the basic objects it holds at their
required rates:
 \begin{equation}
\label{const-bsl}
\forall S_l \in \SS,\quad \sum_{P_u\in \PP}
\sum_{(k,l) \in \download(u)} \dbrate_k \leq \bwprocs_l
\end{equation}

  \item The link between server $S_l$ and processor $P_u$ must have enough
  bandwidth capacity to support all possible object downloads from $S_l$
  to $P_u$ at the required rate:
\begin{equation}
\label{const-bslu}
\forall P_u \in \PP, \forall S_l \in \SS,\quad
\sum_{(k,l) \in \download(u)} \dbrate_k \leq \bwlinks_{l,u}
 \end{equation}

  \item The link between processor $P_u$ and processor $P_v$ must have
  enough bandwidth capacity to support all possible communications between
  the tree nodes mapped on both processors. This constraint can be written
  similarly to constraint~(\ref{const-bpu}) above, but without the
  cost of basic object downloads, and with specifying that $P_u$
  communicates with $P_v$:
\begin{equation}
\label{const-bpuv}
\forall P_u, P_v \in \PP,\quad
 \sum_{j\in\Child(\allocni(u))\cap \allocni(v)} \rho.\delta_j +
     \sum_{j\in \Parent(\allocni(u))\cap \allocni(v)~}\sum_{i\in
\Child(j)\cap \allocni(u)} \rho.\delta_i
       \leq \bwlinkp_{u,v}
\end{equation}

\end{itemize}

\section{Problem Definitions}
\label{sec.probs}

The overall objective of the operator-mapping problem is to ensure that a
prescribed throughput is achieved while minimizing a cost function.  We
consider two broad cases. In the first case, the user must buy processors
(with various computing speed and network card bandwidth specifications)
and build the distributed network dedicated to the application.  For this ``constructive''
problem, which we call \CONS, the cost function is simply the actual
monetary cost of the purchased processors.  This problem is relevant to,
for instance, the surveillance application mentioned in
Section~\ref{sec.intro}. The second case, which we call \NCONS, targets an
existing platform. The goal is then to use a subset of this platform so
that the prescribed throughput is achieved while minimizing a cost
function. Several cost functions can be envisioned, including the compute
capacity or the bandwidth capacity used by the application in steady state,
or a combination of the two. In the following, we consider
a cost function that accounts solely for processors. This function be based
on a processor's processing speed and on the bandwidth of its network card.

Different platform types may be considered for both the \CONS and the
\NCONS problems depending on the heterogeneity of the resources.  In the
\CONS case, we assume that some standard interconnect technology is used to
connect all the processors together ($\bwlinkp_{u,v} = \bwlinkp$). We also
assume that the same interconnect technology is used to connect each server
to processors ($\bwlinks_{l,u} = \bwlinks_l$).
We consider the case
in which the processors are homogeneous because only one type of CPUs and
network cards can be purchased ($\bwprocp_u = \bwprocp$ and $\speed_u =
\speed$). We term the
corresponding problem \CONSHOM. We also consider the case in which the
processors are heterogeneous with various compute speeds and network card
bandwidth, which we term \CONSLAN. In the \NCONS case we consider the case
in which the platform is fully homogeneous, which we term \NCONSHOM. We
then consider the case in which the processors are heterogeneous but the
network links are homogeneous ($\bwlinkp_{u,v} = \bwlinkp$ and $\bwlinks_{l,u} = \bwlinks_{l}$), which we
term \NCONSLAN.
Finally we consider the fully heterogeneous case in which network links
can have various bandwidths, which we term \NCONSHET.

Homogeneity in the platform as described above applies only to
processors and not to servers. Servers are always fixed for a given
application, together with the objects they hold.
We sometimes consider variants of the problem in which the servers and
application tree have particular characteristics. We denote by \HOMS the case
then all servers have identical network capability ($\bwprocs_l=\bwprocs$) and communication links to processors ($\bwlinks_{l,u}=\bwlinks$).
We can also consider the mapping of
particular trees, such as left-deep trees (\LDT) and/or homogeneous
trees with identical object rates $\dbrate_k=\dbrate$ and computing costs
$\calc_i=\calc$ (\HOMA). Also, we can consider application trees with no
communication cost ($\comm_i=0$, \NOCOMA).  All these variants correspond
to simplifications of the problem, and we simply append \HOMS, \LDT, \HOMA,
and/or \NOCOMA to the problem name to denote these simplifications.

\section{Complexity}
\label{sec.np}

Without surprise, most problem instances are NP-hard,
because downloading objects with
different rates on two identical servers is the same problem as
2-Partition~\cite{GareyJohnson}. But from a theoretical point of view,
it is important to assess the complexity of the
simplest instance of the problem,
i.e., mapping a fully homogeneous left-deep
tree application with objects placed on a fully homogeneous set of servers,
 onto a fully homogeneous set of processors:
\CONSHOM-\HOMS-\LDT-\HOMA-\NOCOMA (or \CLDTHom for short).
It turns out that even this problem is difficult, due to the
combinatorial space induced by the mapping of basic objects that are shared
by several operators.  Note that the corresponding non-constructive
problem is exactly the same, since it aims at minimizing the number of
selected processors given a pool of identical processors.  This complexity
result thus holds for both classes of problems.


\begin{definition}
The problem \CLDTHom (\CONSHOM-\HOMS-\LDT-\HOMA-\NOCOMA) consists in
minimizing the number of processors used in the application execution. $K$ is
the prescribed throughput that should not be violated.
\CLDTHomDec is the associated decision problem: given a number of
processors $N$, is there a mapping that achieves throughput~$K$?
\end{definition}

\begin{theorem}
\CLDTHomDec is NP-complete.
\end{theorem}

\begin{proof}
First, \CLDTHomDec belongs to NP. 
Given an allocation of operators to processors and the download
list $\download(u)$ for each processor~$P_u$, we can check in
polynomial time that we use no more than $N$ processors,
that the throughput of each enrolled processor respects $K$:
$$K\times \left| \allocni(u) \right|  \frac{\calc}{\speed}\leq 1\;,$$
and that bandwidth constraints are respected.

To establish the
completeness, we use a reduction from 3-Partition, which is NP-complete
in the strong sense~\cite{GareyJohnson}. We consider an
arbitrary instance $\mathfrak{I}_1$ of 3-Partition:
given $3n$ positive integer numbers $\{ a_{1}, a_{2}, \ldots, a_{3n} \}$
and a bound $R$, assuming that $\frac{R}{4} < a_i <  \frac{R}{2}$ for all $i$
  and that $\sum_{i=1}^{3n} a_i = nR$,
  is there a partition
  of these numbers into $n$ subsets $I_1, I_2, \ldots, I_n$ of sum $R$?
  In other words, are there $n$ subsets $I_1, I_2, \ldots, I_n$ such
that $I_1 \cup I_2 \ldots \cup I_n = \{1, 2, \ldots, 3n\}$,
  $I_i \cap I_j = \emptyset$ if $i \neq j$, and $\sum_{j \in I_i} a_j =
R$ for all $i$ (and $|I_i|=3$ for all~$i$).
 Because 3-Partition is NP-complete in the strong sense, we can encode
the $3n$ numbers in unary and assume that the size of $\mathfrak{I}_1$
is $O(n + M)$, where $M = \max_i \{a_i\}$.

We build the following instance $\mathfrak{I}_2$ of \CLDTHomDec:
\begin{itemize}
  \item The object set is $\OO = \{\obj_1,...,\obj_{3n}\}$, and
there are $3n$ servers each holding an object, thus $\obj_i$ is
available on server~$S_i$. The rate of $\obj_i$ is $\dbrate=1$ , and
the bandwidth limit of the servers is set to~$\bwprocs=1$.
   \item The left-deep tree consists of $|\NN|=nR$ operators with
$\calc=1$.
Each object
 $\obj_i$ appears $a_i$ times in the tree (the exact location does
not matter), so that there are $|\NN|$ leaves in the tree, each
associated to a single operator of the tree.
  \item The platform consists of $n$ processors of speed $\speed
= 1$ and bandwidth $\bwprocp = 3$. All the link bandwidths
interconnecting servers and processors are equal to~$\bwlinks=\bwlinkp=1$.
  \item Finally we ask whether there exists a solution matching the
bounds 
$1/K = R$
and $N=n$.
\end{itemize}

The size of
$\mathfrak{I}_2$ is clearly polynomial in the size of
$\mathfrak{I}_1$, since the size of the tree is bounded by $3nM$.
We now show that instance $\mathfrak{I}_1$ has a solution if and only
if instance $\mathfrak{I}_2$ does.

Suppose first that $\mathfrak{I}_1$ has a solution. We map all
operators corresponding to occurrences of object $\obj_j$, $j \in
I_i$, onto processor $P_i$. Each processor receives three distinct
objects, each coming from a different server, hence bandwidths
constraints are satisfied. Moreover, the number of operators computed by
$P_i$ is equal to $\sum_{j\in I_i} a_i = R$, and the required throughput it achieved
because $K R \leq 1$.
We have thus built a solution to $\mathfrak{I}_2$.

Suppose now that $\mathfrak{I}_2$ has a solution, \ie a mapping
matching the bound 
$1/K = R$
with $n$~processors.
Due to bandwidth constraints, each of the $n$ processors is assigned
at most three
distinct objects. Conversely, each object must be assigned to at least
one processor and there are $3n$ objects, so each processor is
assigned exactly $3$ objects in the solution, and no object is sent to
two distinct processors. Hence, a processor must compute all operators
corresponding to the objects it needs to download, which directly
leads to a solution of $\mathfrak{I}_1$ and concludes the proof.
\end{proof}

Note that problem \CLDTHomDec becomes polynomial if one adds the additional
restriction that no basic object is used by more than one operator in the
tree. In this case, one can simply assign operators to $\lceil |\NN| \times \calc
/ \speed \rceil$ arbitrary processors in a round-robin fashion.

\section{Linear Programming Formulations}
\label{sec.ilp}

In this section, we formulate the \CONS optimization problem
as an integer linear program (ILP). We deal with the most general
instance of the problem \CONSLAN.
Then we explain how to transform this integer linear program to formulate
the \NCONSHET problem.

\subsection{ILP for \CONS}$ $

{\bf Constants --} We first define the set of constant values that
define our problem.  The application tree is defined via parameters $par$ and $leaf$, and
the location of objects on servers is defined via parameter $obj$.
Other parameters are defined with the same notations as
previously introduced: $\comm_i, \calc_i$ for operators,
$\dbrate_k$ for object download rates, and $\bwprocs_l$ for server network card bandwidths.
More formally:
\begin{itemize}
\item $par(i,j)$ is a boolean variable equal to $1$ if operator
$n_i$ is the parent of $n_j$ in the application tree, and $0$
otherwise.
\item $leaf(i,k)$ is a boolean variable equal to $1$ if operator $n_i$
requires object $\obj_k$ for computation, i.e., $o_k$ is a children of
$n_i$ in the tree. Otherwise $leaf(i,k)=0$.
\item $obj(k,l)$ is a boolean variable equal to $1$ if server $S_l$
holds a copy of object~$\obj_k$.
\item $\comm_i, \calc_i, \dbrate_k, \bwprocs_l$ are rational numbers.
\end{itemize}

The platform can be built using different types of processors. More formally,
we consider a set $\CC$ of processor specifications, which we call ``classes''. We can
acquire as many processors of a class $c \in \CC$ as needed, although no more
than $\NN$ processors are necessary overall.
We denote the cost of a processor in class $c$ by $cost_c$. Each processor of class~$c$ has
computing speed $\speed_c$ and network card bandwidth $\bwprocp_c$. The link bandwidth between processors is
a constant $\bwlinkp$, while the link between a server~$S_l$ and a processor
is~$\bwlinks_{l}$. For each class, processors are
numbered from~$1$ to $|\NN|$, and $P_{c,u}$ refers to the $u^{th}$
processor of class~$c$. Finally, $\rho$ is the throughput that must be achieved
by the application:
\begin{itemize}
\item $cost_c, \speed_c, \bwprocp_c, \bwlinkp, \bwlinks_{l}$ are
rational numbers;
\item $\rho$ is a rational number.
\end{itemize}

\medskip
{\bf Variables --} Now that we have defined the constants that define our problem we define unknown variables to be computed:
\begin{itemize}
\item $x_{i,c,u}$ is a boolean variable equal to $1$ if operator $n_i$ is
mapped on $P_{c,u}$, and $0$ otherwise.
There are $|\NN|^2.|\CC|$ such variables, where $|\CC|$ is the number
of different classes of processors.
\item $d_{c,u,k,l}$ is a boolean variable equal to $1$ if processor
$P_{c,u}$ downloads object $\obj_k$ from server~$S_l$, and $0$
otherwise. The number of such variables is $|\CC|.|\NN|.|\OO|.|\SS|$.
\item $y_{i,c,u,i',c',u'}$ is a boolean variable equal to $1$ if $n_i$
is mapped on $P_{c,u}$, $n_{i'}$ is mapped on
$P_{c',u'}$, and $n_i$ is the parent of $n_{i'}$
in the application tree. There are $|\NN|^4.|\CC|^2$ such variables.
\item $used_{c,u}$ is a boolean variable equal to $1$ if processor
$P_{c,u}$ is used in the final mapping, i.e., there is at
least one operator mapped on this processor, and $0$ otherwise. There are
$|\CC|.|\NN|$ such variables.
\end{itemize}

\medskip
{\bf Constraints --} Finally, we must write all constraints involving
our constants and variables.  In the following, unless stated otherwise,
$i$, $i'$, $u$ and $u'$ span set~$\NN$; $c$ and $c'$ span set~$\CC$; $k$
spans set~$\OO$; and $l$ spans set~$\SS$.
First we need constraints to guarantee that the allocation of operators
to processors is a valid allocation, and that all required downloads
of objects are done from a server that holds the corresponding object.

\begin{itemize}
\item $\forall i~~ \sum_{c,u} x_{i,c,u}=1$: each operator is
placed on exactly one processor;
\item $\forall c,u,k,l~~ d_{c,u,k,l} \leq obj(k,l)$: object~$\obj_k$
can be downloaded from $S_l$ only if $S_l$ holds $\obj_k$;
\item $\forall c,u,k,l~~ d_{c,u,k,l} \leq \sum_i x_{i,c,u}.leaf(i,k)$:
if there is no operaotr assigned to $P_{c,u}$ that
requires object $k$, then $P_{c,u}$ does not need to download
object~$k$ and $d_{c,u,k,l}=0$ for all server~$S_l$.
\item $\forall i,k,c,u~~ 1\geq \sum_l d_{c,u,k,l} \geq
x_{i,c,u}.leaf(i,k)$: processor $P_{c,u}$ must download
object~$\obj_k$ from exactly one server if there is an operator $n_i$
mapped on this processor that requires $\obj_k$ for computation.
\end{itemize}

\medskip

The next set of constraints aim at properly constraining
variable~$y_{i,c,u,i',c',u'}$. Note that a straightforward definition would be
$y_{i,c,u,i',c',u'}=par(i,j).x_{i,c,u}.x_{i',c',u'}$, i.e., a logical conjunction
between three conditions. Unfortunately, this definition makes our program non-linear as two of the conditions are variables.
Instead, for all $i,c,u,i',c',u'$, we write:
\begin{itemize}
\item $y_{i,c,u,i',c',u'} \leq par(i,j)$; $~y_{i,c,u,i',c',u'} \leq
x_{i,c,u}$; $~y_{i,c,u,i',c',u'} \leq x_{i',c',u'}$: $y$ is forced to
$0$ if one of the conditions does not hold.
\item $y_{i,c,u,i',c',u'} \geq par(i,j).\left( x_{i,c,u} +
x_{i',c',u'}-1 \right)$: $y$ is forced to be $1$ only if
the three conditions are true (otherwise the right term is less than or
equal to~$0$).
\end{itemize}

\medskip

The following constraints ensure that $used_{c,u}$ is properly defined:
\begin{itemize}
\item $\forall c,u~ used_{c,u} \leq \sum_{i} x_{i,c,u}$:
processor $P_{c,u}$ is not used if no operator is mapped on it;
\item $\forall c,u,i~ used_{c,u} \geq x_{i,c,u}$: processor $P_{c,u}$ is
used if at least one operator~$n_i$ is mapped to it.
\end{itemize}

\medskip

Finally, we have to ensure that the required throughput is achieved and that
the various bandwidth capacities are not exceeded, following equations
(\ref{const-comp})-(\ref{const-bpuv}).

\begin{itemize}
\item $\forall c,u~~ \sum_{i} x_{i,c,u}.\rho
\frac{\calc_i}{\speed_c} \leq 1$: the computation of each processor must
be fast enough so that the throughput is at least equal to $\rho$;


\item $\forall c,u~~  \sum_{k,l}
d_{c,u,k,l}.\dbrate_k~+~\sum_{i,i',(c',u')\neq(c,u)}~y_{i,c,u,i',c',u'}.\rho.\comm_{i'}
+~\sum_{i,i',(c',u')\neq(c,u)}~y_{i',c',u',i,c,u}.\rho.\comm_{i}
\leq \bwprocp_c$: bandwidth constraint for the
processor network cards;

\item $\forall l~~ \sum_{c,u,k}~ d_{c,u,k,l}.\dbrate_k \leq
\bwprocs_l$: bandwidth constraint for the server network cards;

\item $\forall l,c,u~~  \sum_k d_{c,u,k,l}.\dbrate_k \leq
\bwlinks_{l}$: bandwidth constraint for links between servers and
processors;

\item $\forall c,u,c',u'$ with $(c,u)\neq (c',u')~~
\sum_{i,i'} y_{i,c,u,i',c',u'}.\rho.\comm_{i'} +
\sum_{i,i'} y_{i',c',u',i,c,u}.\rho.\comm_{i} \leq \bwlinkp$:
bandwidth constraint for links between processors.
\end{itemize}






\medskip
{\bf Objective function.}

We aim at minimizing the cost of used processors, thus the objective
function is
\begin{center}
$\min \left( \sum_{c,u} used_{c,u}.cost_c \right)$.
\end{center}

\medskip

\subsection{ILP for \NCONS}
The linear program for the \NCONS problem is very similar to the \CONS
one, except that the platform is known a-priori. Furthermore, we no longer
consider processor classes.  However, we can simply assume that there is only one
processor of each class, and define~$|\CC|=|\PP|$, the set of processors
of the platform. The {\em number} of processors of
class~$c$ is then limited to~$1$. As a result, all indices~$u$ in
the previous linear program are removed, and we obtain
a linear program formulation of the \NCONSLAN problem.
The number of variables and
constraints is reduced from $|\NN|$ to $1$ when appropriate. 
We can further generalize the linear program to \NCONSHET, by adding
links of different bandwidths between processors. We just need to replace
$\bwlinkp$ by $\bwlinkp_{c,c'}$ and $\bwlinks_l$ by $\bwlinks_{l,c}$ every
time they appear in the linear program in the previous section.
Altogether, we have provided integer linear program formulations for
all our constructive and non-constructive problems.

\section{Heuristics}
\label{sec.heuristics}


In this section we propose several heuristics to solve the \CONS
operator-placement problem. Due to lack of space, we leave the development
of heuristics for the \NCONS problem outside the scope of this
paper. We choose to focus on constructive scenarios because such scenarios
are relevant to practice and, to the best of our knowledge, have not been
studied extensively in the literature. We say that the heuristics
can then ``purchase'' processors, or ``sell back'' processors, until a
final set of needed processors is determined.

We consider two types of heuristics: (i)~operator placement heuristics and
(ii)~object download heuristics. In a first step, an operator placement
heuristic is used to determine the number of processors that should be
purchased, and to decide which operators are assigned to which processors.
Note that all our heuristics fail if a single operator cannot be treated by
the most expensive processor with the desired throughput.  In a second
step, an object download heuristic is used to decide from which server each
processor downloads the basic objects that are needed for the operators
assigned to this processor. In the next two sections we propose several
candidate heuristics both for operator placement and object download.

\subsection{Operator Placement Heuristics}

\subsubsection{Random}

While there are some unassigned operators, the Random heuristic picks one of
these unassigned operators randomly, called $op$. It then purchases the
cheapest possible processor that is able to handle $op$ while achieving the
required application throughput. If there is no such processor, then the
heuristic considers $op$ along with one of its children operators or with
its parent operator. This second operator is chosen so that it has the most
demanding communication requirements with $op$ (the intuition
is that we try to reduce communication overhead).
If no processor can be acquired
that can handle both the operators together, then the heuristic fails. If
the additional operator had already been assigned to another processor,
this last processor is sold back.

\subsubsection{Comp-Greedy}

The Comp-Greedy heuristic first sorts operators in non-increasing order of
$\calc_i$, i.e., most computationally demanding operators first. While
there are unassigned operators, the heuristic purchases the most expensive
processor available and assigns the most computationally demanding unassigned operator
to it. If this operator cannot be processed on this processor so that the
required throughput is achieved, then the heuristic uses a grouping
technique similar to that used by the Random heuristic (i.e., trying to
group the operator with its child or parent operator with which it has the
most demanding communication requirement). If after this step some capacity
is left on the processor, then the heuristic tries to assign other
operators to it. These operators are picked in non-increasing order of
$\calc_i$, i.e., trying to first assign to this processor the most
computationally demanding operators. Once no more operators can be
assigned to the processor, the heuristic attempts to ``downgrade'' the
processor. This downgrading consists in, if possible, replacing the current
processor by the cheapest processor available that can still handle all the
operators assigned on the current processor.

\subsubsection{Comm-Greedy}

The Comm-Greedy heuristic attempts to group operators to reduce
communication costs.  It picks the two operators that
have the largest communication requirements.
These two operators are grouped and assigned to the same processor,
thus saving the costly communication between both processors. There
are three cases to consider for this assignment: (i)~both operators were
unassigned, in which case the heuristic simply purchases the cheapest
processor that can handle both operators; if no such processor is available
then the heuristic purchases the most expensive processor for each operator;
(ii)~one of the operators was already assigned to a processor, in which
case the heuristic attempts to accommodate the
other operator as well; if this is not possible then the heuristic
purchases the most expensive processor for the other operator; (iii)~both
operators were already assigned on two different processors, in which case
the heuristic attempts to accommodate both operators on one processor and
sell the other processor; if this is not possible then the current operator
assignment is not changed.

\subsubsection{Object-Greedy}

The Object-Greedy heuristic attempts to group operators that need the same
basic objects. Recall that an al-operator is an operator that requires at
least one basic object. The heuristic sorts all al-operators by the maximum
required download frequency of the basic objects they require, i.e., in
non-increasing order of maximum $rate_j$ values  (and $\calc_i$ in case of
equality). The heuristic then purchases the most expensive processor and
assigns the first such operators to it. Once again, if the most expensive
processor cannot handle this operator, the heuristic attempts to group
the operator with one of its unassigned parent or child operators. If this
is not possible, then the heuristic fails. Then, in a greedy fashion, this
processor is filled first with al-operators and then with other operators as much
as possible.

\subsubsection{Subtree-Bottom-Up}

The Subtree-Bottom-Up heuristic first purchases as many most expensive
processors as there are al-operators and assigns each al-operator to a
distinct processor.  The heuristic then tries to merge the operators with
their father on a single machine, in a bottom-up fashion (possibly leading
to the selling back of some processors). Consider a processor on which a
number of operators have been assigned. The heuristic first tries to
allocate as many parent operators of the currently assigned operators to
this processor. If some parent operators cannot be assigned to this
processor, then one or more new processors are purchased.  This mechanism
is used until all operators have been assigned to processors.

\subsubsection{Object-Grouping}

For each basic object, this heuristic counts how many operators need this
basic object. This count is called the ``popularity'' of the basic object.
The al-operators are then sorted by non-increasing sum of the popularities
of the basic object they need. The heuristic starts by purchasing the most
expensive processor and assigning to it the first al-operator. The
heuristic then attempts to assign as many other al-operators that require
the same basic objects as the first al-operator, taken in order of
non-increasing popularity, and then as many non al-operators as possible.
This process is repeated until all operators have been assigned.

\subsubsection{Object-Availability}
This heuristic takes into account the distribution of basic objects on
the servers. For each object $k$ the number $av_k$ of servers handling
object~$o_k$ is
calculated. Al-operators in turn are treated in increasing order of
$av_k$ of the basic objects they need to download. The heuristic tries
to assign as many al-operators downloading object $k$ as possible on
a most expensive processor. The remaining internal operators are
assigned in the same mechanism as Comp-Greedy proceeds, i.e., in
decreasing order of $w_i$ of the operators.

\subsection{Object Download Heuristics}

Once an operator placement heuristic has been executed, each al-operator is
mapped on a processor, which needs to download basic objects required by
the operator. Thus, we need to specify from which server this download should
occur. Two server selection heuristics are proposed in order to define, for
each processor, the server from which required basic objects are
downloaded.


\subsubsection{Server-Selection-Random}
This heuristic is only used in combination with Random.
Once Random has decided about the mapping of
operators onto processors, Server-Selection-Random associates randomly a
server to each basic object a processor has to download.

\subsubsection{Server-Selection-Intelligent}
This server selection heuristic is more sophisticated and is used in
combination with all operator placement heuristics except Random.
Server-Selection-Intelligent uses three loops: the first
loop assigns objects that are held in exclusivity, i.e., objects that
have to be downloaded from a specific server. If not all downloads can
be guaranteed, the heuristic fails. The second loop
associates as many downloads as possible to servers that provide only
one basic object type. The last loop finally tries to assign the
remaining basic objects that have to be downloaded. For this purpose
objects are treated in decreasing order of
interestedProcs/numPossibleServers, where interestedProcs is the
remaining number of processors that need to download the object and
numPossibleServers is the number of servers where the object still can
be downloaded. In the decision process servers are considered in
decreasing order of min(remainingBW, linkBW),
where remainingBW is the remaining capacity of the servers network
card and linkBW is the bandwidth of the communication link.

\medskip
Once the server association process is done, a processor downgrade
procedure is called. All processors are replaced by the less expensive
model that fulfills the CPU and network card requirements of the
allocation.

\section{Simulation Results}
\label{sec.results}

\subsection{Resource Cost Model}

In order to instantiate our simulations with realistic models for resource
costs, we use information available from the Dell Inc. Web site. More
specifically, we use the prices for configurations of Intel's latest,
high-end, rack-mountable server (PowerEdge R900), as advertised on the Web
site as of early March 2008. Due to the large number of available configurations,
we only consider processor cores with 8MB L1 caches (so that their
performances are more directly comparable), and with optical Gigabit Ethernet (GbE)
network cards manufactured by Intel Inc. For simplicity, we assume that the
effective bandwidth of a network card is equal to its peak performance. In reality,
we know that, say, a 10GbE network card delivers a bandwidth lower
than 10Gbps due to various software and hardware overheads.  We also make
the assumption that the performance of a multi-processor multi-core server
is proportional to the sum of the clock rates of all its cores. This
assumption generally does not hold in practice due, e.g., to
parallelization overhead and cache sharing. It is outside the scope of this
work to develop (likely elusive) generic performance models for network
cards and multi-processor multi-core servers, but we argue that the above
assumptions still lead to a reasonable resource cost model. The
configuration prices are show in Table~\ref{tab.conf_prices}, relative to
the base configuration, whose cost is \$7,548.  Note that we do not
consider configurations designed for low power consumption, which
achieve possibly lower performance at higher costs.

\begin{center}
\begin{table}[Ht]
\caption{Incremental costs for increases in processor performance or
network card bandwidth relative to a \$7,548 base configuration (based on data
from the Dell Inc. web site, as of early March 2008).}
\label{tab.conf_prices}
\begin{center}
\begin{tabular}{|c|c|c||c|c|c|}
\hline
\multicolumn{3}{|c||}{Processor} & \multicolumn{3}{|c|}{Network Card}\\
\hline
Performance  & Cost & Ratio & Bandwidth & Cost & Ratio\\
(GHz) & (\$) & (GHz/\$) & (Gbps) & (\$) & (Gbps/\$) \\
\hline
\hline
11.72	& 7,548 + 0     & 1.55 $\times 10^{-3}$ & 1   & 7,548 + 0      & 1.32 $\times 10^{-4}$\\
19.20	& 7,548 + 1,550 & 1.93 $\times 10^{-3}$ & 2   & 7,548 + 399    & 2.51 $\times 10^{-4}$\\
25.60	& 7,548 + 2,399 & 2.38 $\times 10^{-3}$ & 4   & 7,548 + 1,197  & 4.57 $\times 10^{-4}$\\
38.40	& 7,548 + 3,949 & 3.12 $\times 10^{-3}$ & 10  & 7,548 + 2,800  & 9.66 $\times 10^{-4}$\\
46.88	& 7,548 + 5,299 & 3.43 $\times 10^{-3}$ & 20  & 7,548 + 5,999  & 14.76 $\times 10^{-4}$\\
\hline
\end{tabular}
\end{center}
\end{table}
\end{center}

\subsection{Simulation Methodology}
\label{sec.simumethod}

All our simulations use randomly generated binary operator trees with at
most $N$ operators, which can be specified.  All leaves correspond to basic
objects, and each basic object is chosen randomly among 15 different types.
For each of these 15 basic object types, we randomly choose a fixed size.
In simulations with \emph{small objects}, the object sizes are in the range
5-30MB, whereas \emph{big objects} have data sizes in the range 450-530MB.
The download frequency for basic objects is either fixed to 1/50s or 1/2s.
The computation amount $w_n$ for an operator $n$ (a non-leaf node in the
tree), depends on its children $l$ and $r$: $w_n = (\delta_l +
\delta_r)^{\alpha}$, where $\alpha$ is a constant fixed for each simulation
run. The same principle is used for the output size of each operator, using
a constant $\beta=1.0$ for all simulations.  The application throughput
$\rho$ is fixed to $1.0$ for all simulations. Throughout the whole set of
simulations we use the same server architecture: we dispose of 6 servers,
each of them is equipped with a 10 GB network card. Objects of our 15 types
are randomly distributed over the 6 servers. We assume that servers and
processors are all interconnected by a 1GB link. The mapping operator
problem is defined by many parameters, an we argue that our simulation
methodology, in which several parameters are fixed, is sufficient to
compare our various heuristics.

\subsection{Results}
\label{sec.simu}

We present hereafter results for several sets of experiments.  Due to lack
of space we will only present the most significant figures, but the entire
set of figures can be found on the web~\cite{codeVeroQuery}.

\paragraph{High download rates - small object sizes}
In a first set of simulations, we study the behavior of the heuristics when
download rates are high and object sizes small (5-30MB).
Figure~\ref{fig.exp1} shows the results, when the number of nodes $N$ in
the tree varies, but the computation factor $\alpha$ is fixed. As expected,
Random performs poorly and the platform chosen for an application with
around 100 operators or more exceeds a cost of \$400,000 (cf.
Figure~\ref{fig.exp1-0.5}), when $\alpha = 0.5$). Subtree-bottom-up
achieves the best costs, and for an application with 100 operators it finds
a platform for the price of \$8,745. All Greedy heuristics exhibit similar
performance, slightly poorer than Subtree-bottom-up, but still withing
acceptable costs under \$50,000. Perhaps surprisingly, the heuristics that
pay special attention for basic objects, Object-Grouping and
Object-Availability, perform poorly.

With a larger value of $\alpha$ (cf. Figure~\ref{fig.exp1-1.7}) the
operator tree size becomes a more limiting factor. For trees with more than
80 operators, almost no feasible mapping can be found.  However, the
relative performance of our heuristics remains almost the same, with two
notable features: a)~Object-Grouping still finds some mappings for
operator trees bigger up to 120 operators, with costs between \$200,000 and
\$275,000; b)~Comp-Greedy and Object-Greedy perform as well at at times
better than Subtree-bottom-up when the number of operator increases.

\begin{figure}
   \centering
   \subfigure[$\alpha = 0.9$.]{
     \includegraphics[angle=270, width=0.45\textwidth]{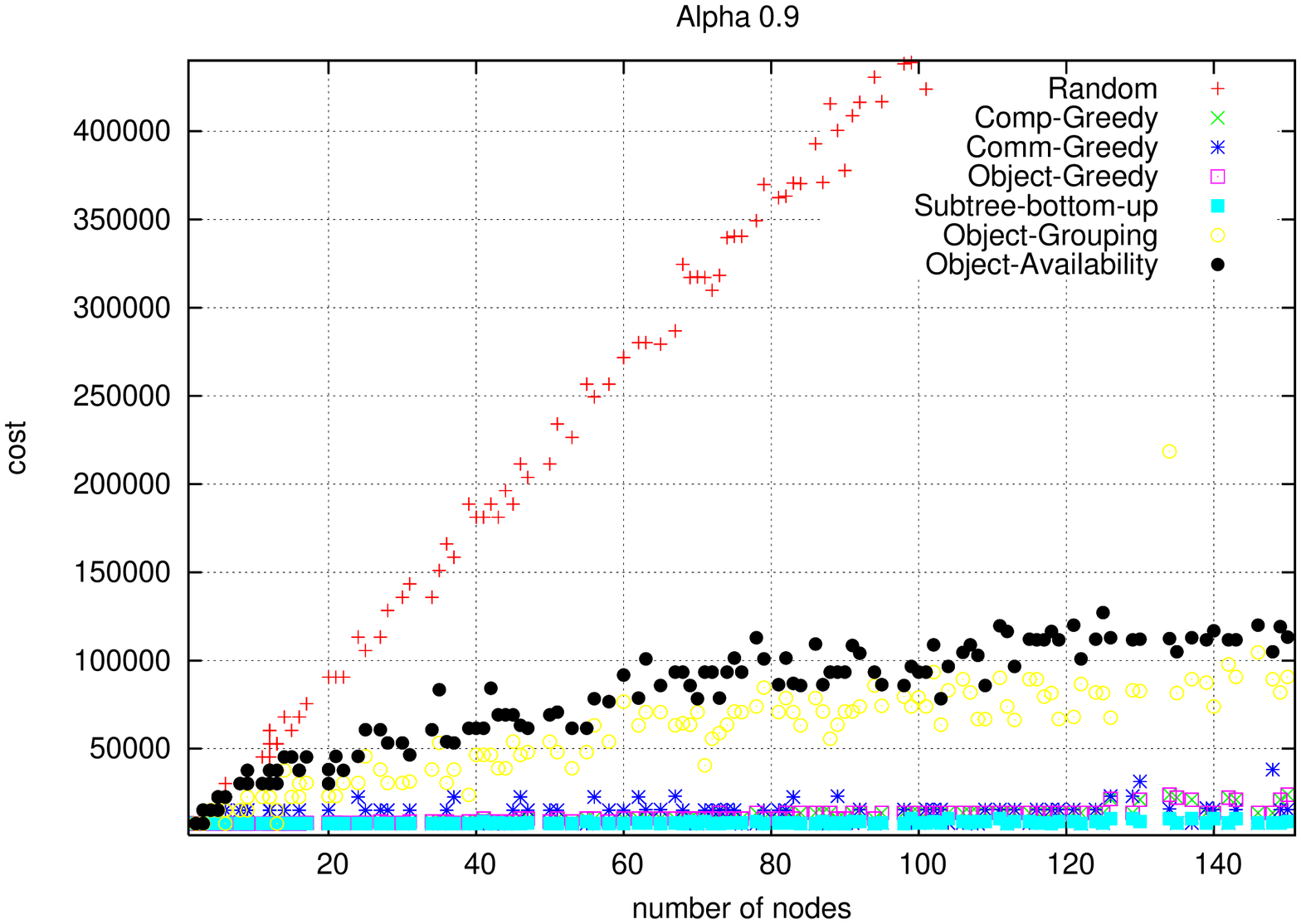}
     \label{fig.exp1-0.5}
   }$\quad$
   \subfigure[$\alpha = 1.7$.]{
     \includegraphics[angle=270, width=0.45\textwidth]{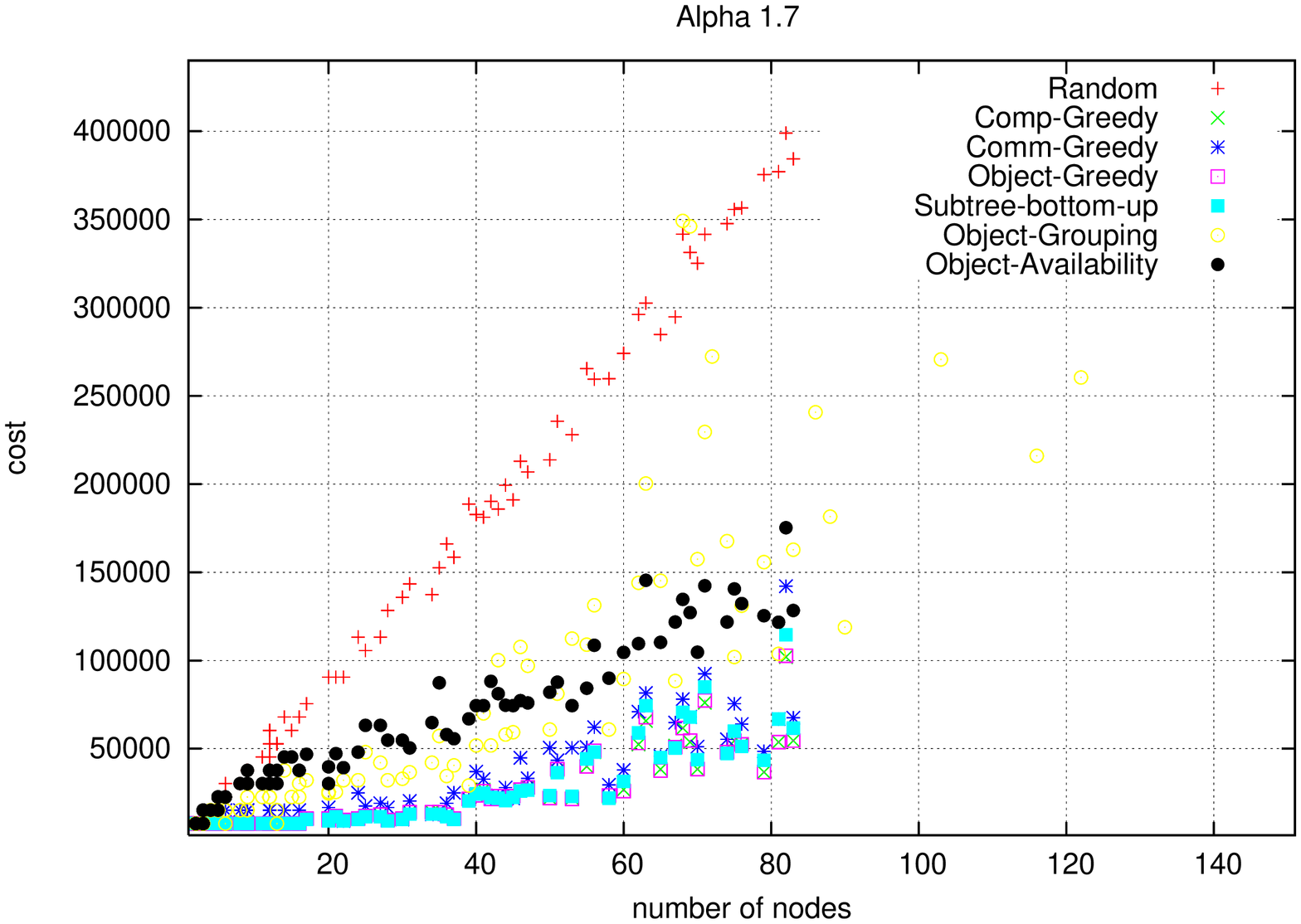}
     \label{fig.exp1-1.7}
   }
   \caption{Simulation with small basic objects and big download
   rates, increasing number of operators.}
\label{fig.exp1}
\end{figure}

Figure~\ref{fig.exp2} shows the comparison of the heuristics when
 $N$ is fixed and the computation factor $\alpha$ increases. This
 experiment uses the same parameters as the previous one. Up to a threshold
 the $\alpha$ parameter has no influence on the heuristics' performance
 and the solution cost is linear.  When $\alpha$ reaches the threshold, the
 solution cost of each heuristic increases until $\alpha$ exceeds a second
 threshold and no solution can be found anymore. Depending on the number
 of operators both thresholds have lower or higher values. In the case of
 small operator trees with only 20 nodes (see Figure~\ref{fig.exp2-20}),
 the first threshold is for $\alpha=1.7$ and the second at $\alpha = 2.2$
 (vs. $\alpha = 1.6$ and $\alpha = 1.8$ for operator trees of size 60, as
 seen in Figure~\ref{fig.exp2-60}).  Subtree-bottom-up behaves in both
 cases the best, whereas Random performs the poorest. Object-Grouping and
 Object-Availability change their position in the ranking: for small trees
 Object-Grouping behaves better, while for bigger trees it is outperformed
 by Object-Availability. The Greedy heuristics are between
 Subtree-bottom-up and the object sensitive heuristics. When $\alpha$ is
 larger, they at times outperform Subtree-bottom-up.

\begin{figure}
   \centering
   \subfigure[$N = 20$.]{
     \includegraphics[angle=270, width=0.45\textwidth]{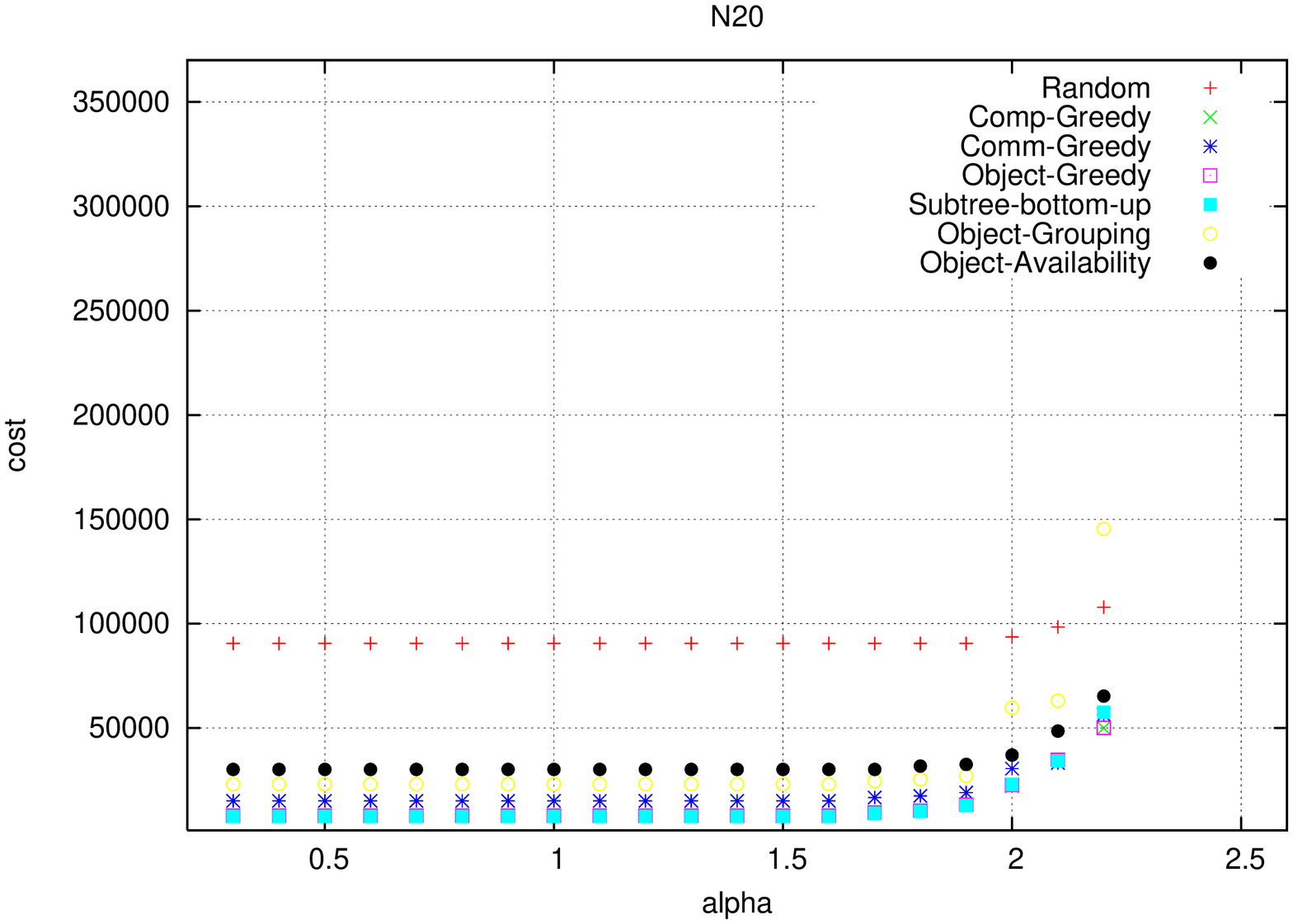}
     \label{fig.exp2-20}
   }$\quad$
   \subfigure[$N = 60$.]{
     \includegraphics[angle=270, width=0.45\textwidth]{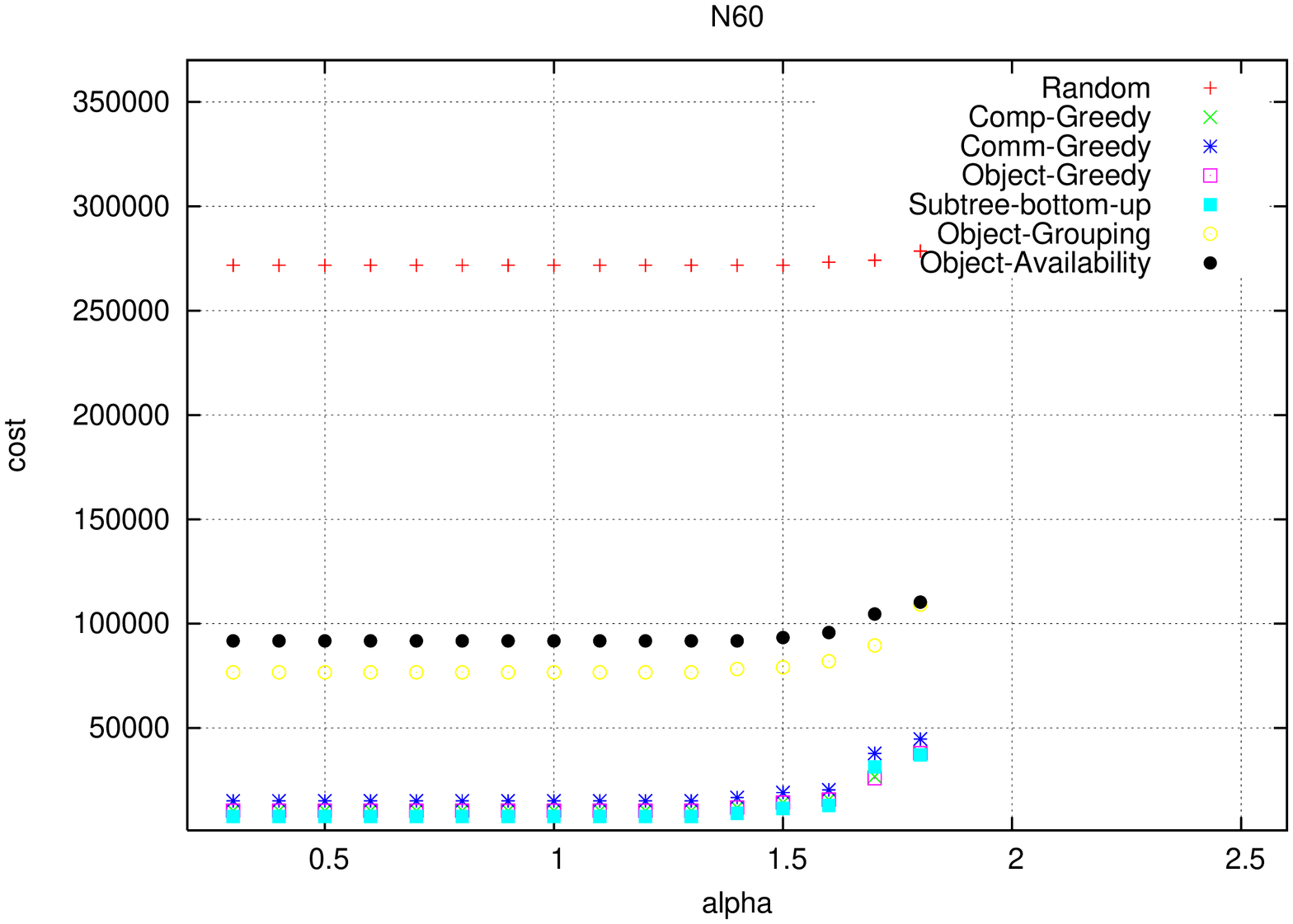}
     \label{fig.exp2-60}
   }
   \caption{Simulation with small basic objects and big download
   rates, increasing $\alpha$.}
\label{fig.exp2}
\end{figure}

\paragraph{High download rates - big object sizes}
The second set of experiments analyzes the heuristics' performance under
high download rates and big object sizes (450-530MB). As for small object
sizes, we plot two types of figures. Figure~\ref{fig.exp5} shows results
for a fixed $\alpha$ and increasing number of operators. We see that for
trees bigger than 45 nodes, almost no feasible solution can be found, both
for $\alpha$ smaller than 1 and higher than 1. In general,
Subtree-bottom-up still achieves the best costs, but at times it is
outperformed by Comm-Greedy. Subtree-bottom-up even fails in two cases in
which other heuristics find a solution: see Figure~\ref{fig.exp5-0.5}, N=41
and N=42. This behavior can be explained as follows.  The Subtree-bottom-up
routine achieves the best result in terms of processors that have to be
purchased. But unfortunately this operator-processor-mapping fails during
the server allocation process. (Often the bandwidth of 1 GB between
processor and server is not sufficient).

Comm-Greedy achieves in this experiment the best costs among the Greedy
heuristics, whereas Random, Object-Availability and Object-Grouping still
perform the poorest.

\begin{figure}
   \centering
   \subfigure[$\alpha = 0.9$.]{
     \includegraphics[angle=270, width=0.45\textwidth]{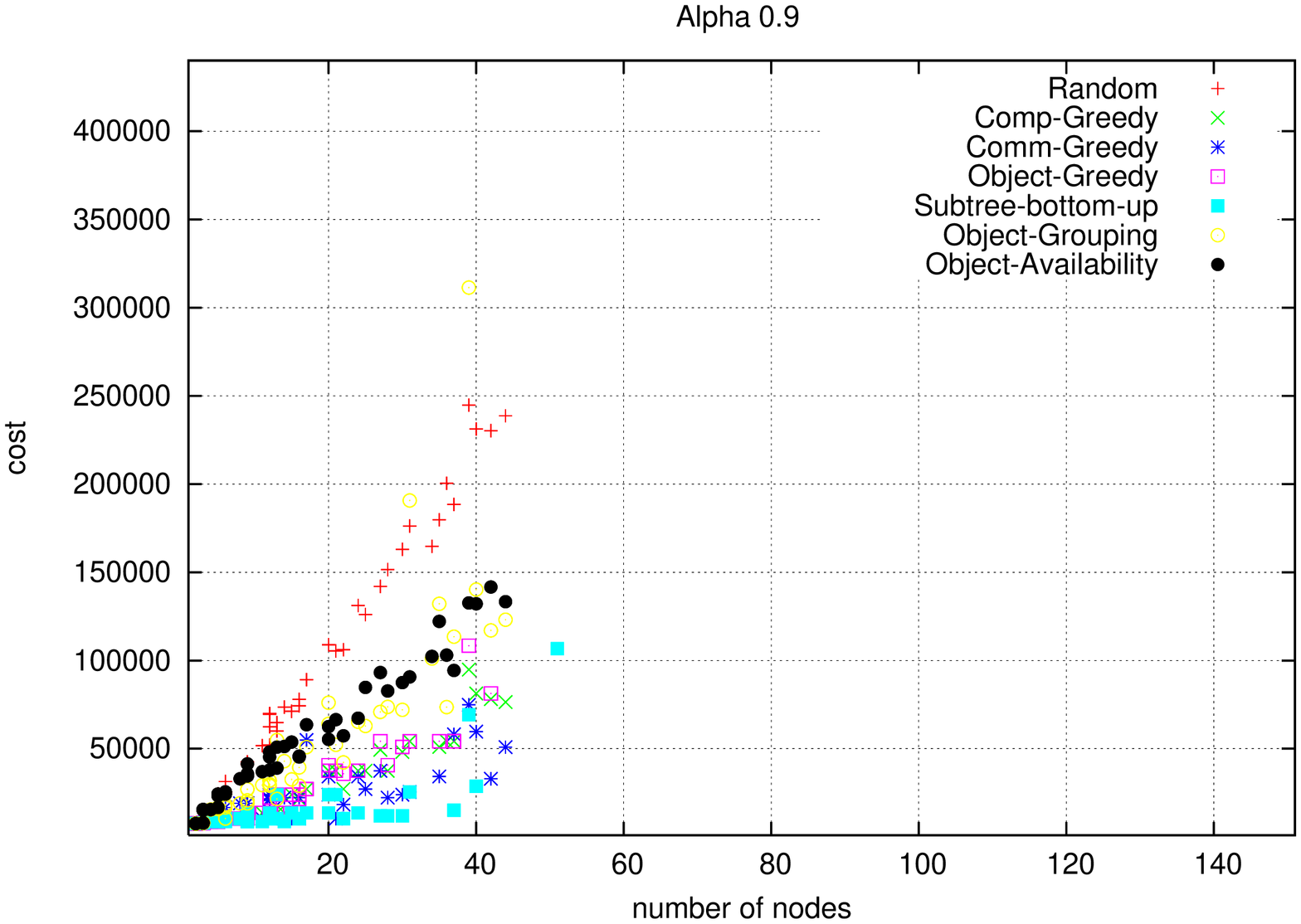}
     \label{fig.exp5-0.5}
   }$\quad$
   \subfigure[$\alpha = 1.1$.]{
     \includegraphics[angle=270, width=0.45\textwidth]{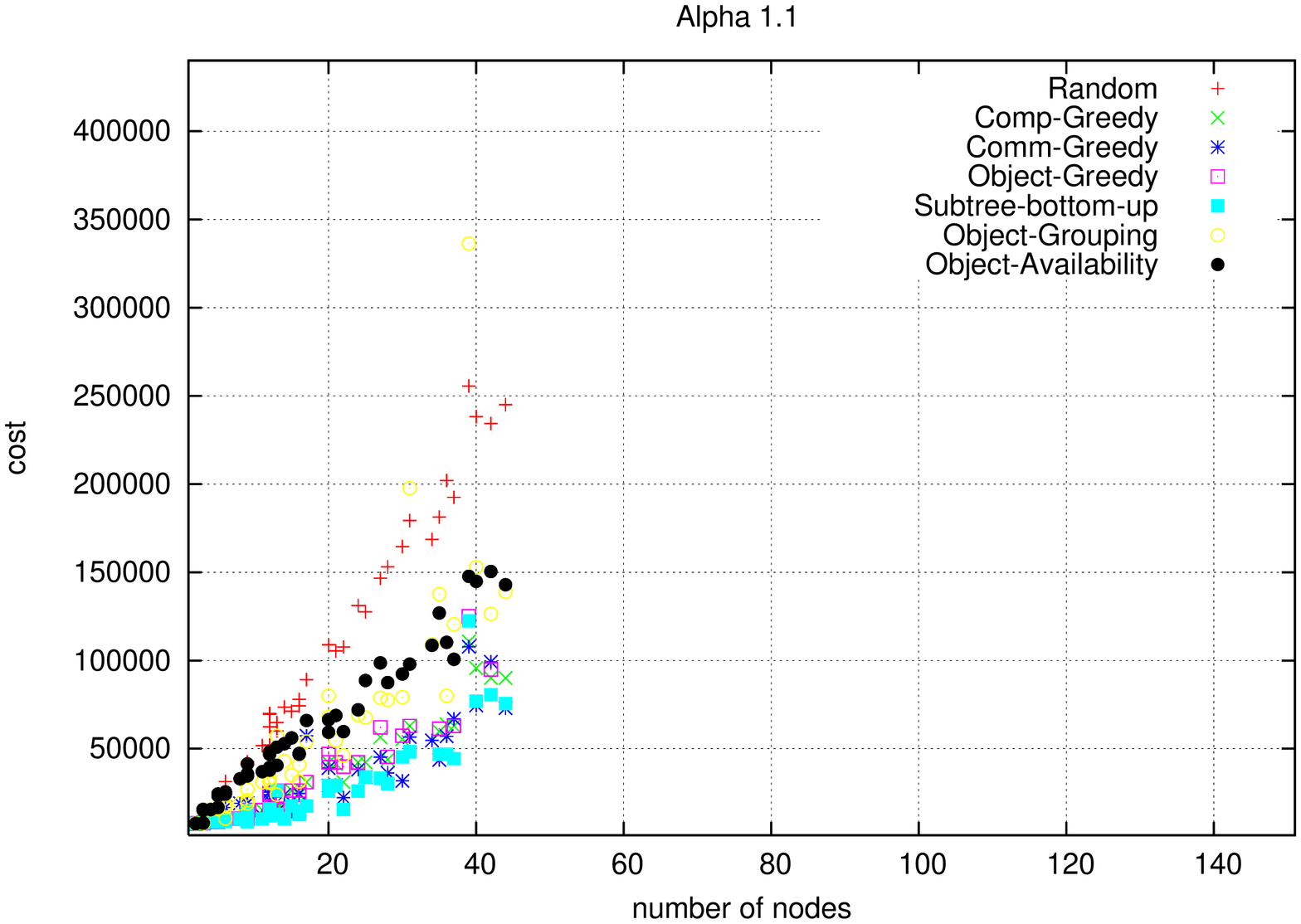}
     \label{fig.exp5-1.7}
   }
   \caption{Simulation with big basic objects and high download
   rates, increasing number of operators.}
\label{fig.exp5}
\end{figure}

When $N$ is fixed we observe a behavior similar as that for small object
sizes.  The ranking (Subtree-bottom-up, Greedy, object sensitive, and
finally Random) remains unchanged. When $N=20$, Comp-Greedy outperforms
Object-Greedy and Comm-Greedy finds a feasible
solution only once (see Figure~\ref{fig.exp6-20}). Object-Availability achieves
better results than Object-Grouping.

In the case of $N=40$ (see Figure~\ref{fig.exp6-40}), the ranking is
unchanged but for the fact that Object-Availability and Object-Grouping are
swapped. Also, in this case, Object-Greedy never succeeds to find a
feasible solution, whereas Comm-Greedy achieves the second best results.

Note that the failure of Object-Greedy depends on the
tree structure, and thus our results do not mean that
Object-Greedy fails for all
tree sizes higher than 20. Here once again, the solution found by the
heuristic for the operator mapping leads to the failure in the server
association process.
%

\begin{figure}
   \centering
   \subfigure[$N = 20$.]{
     \includegraphics[angle=270, width=0.45\textwidth]{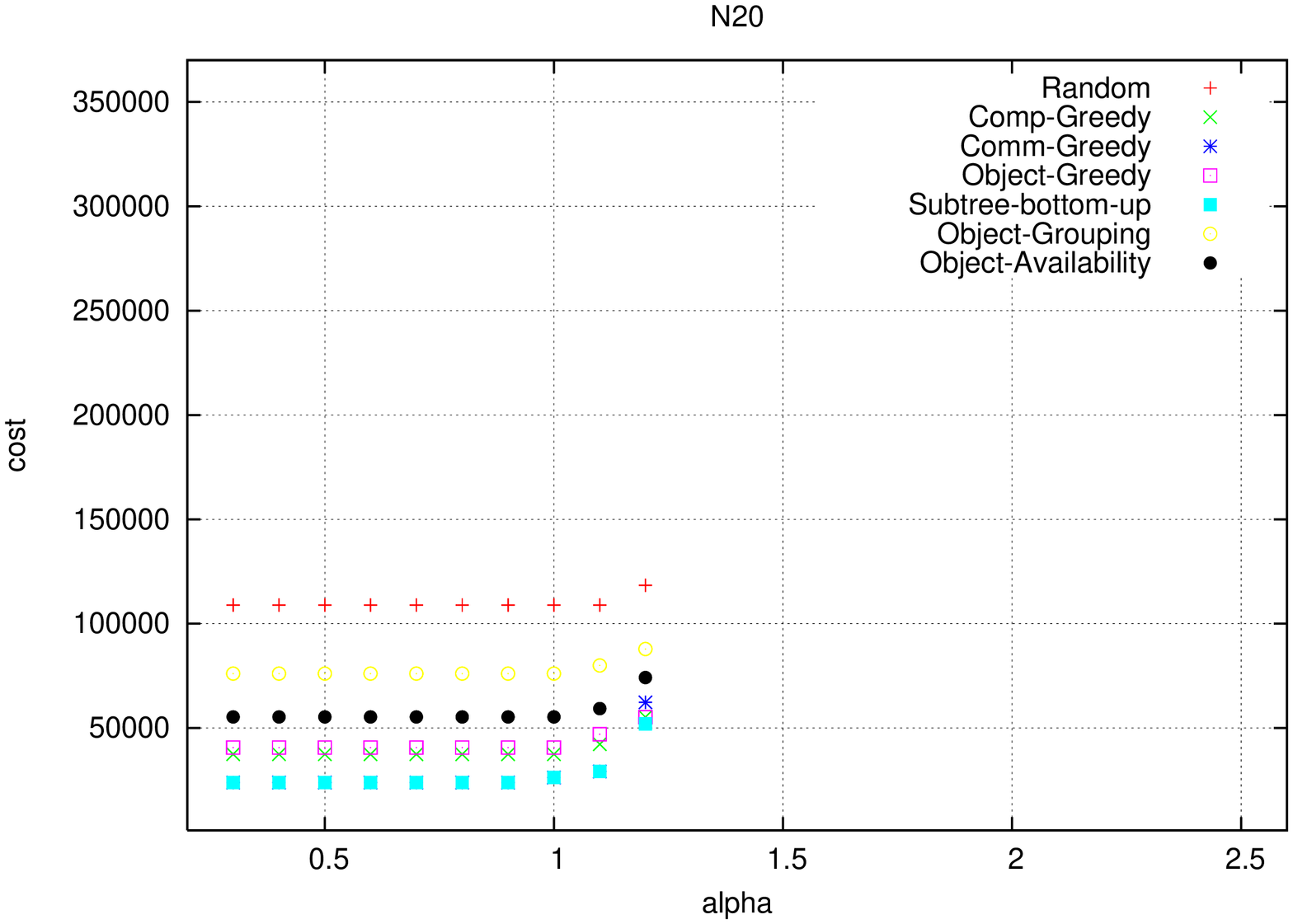}
     \label{fig.exp6-20}
   }$\quad$
   \subfigure[$N = 40$.]{
     \includegraphics[angle=270, width=0.45\textwidth]{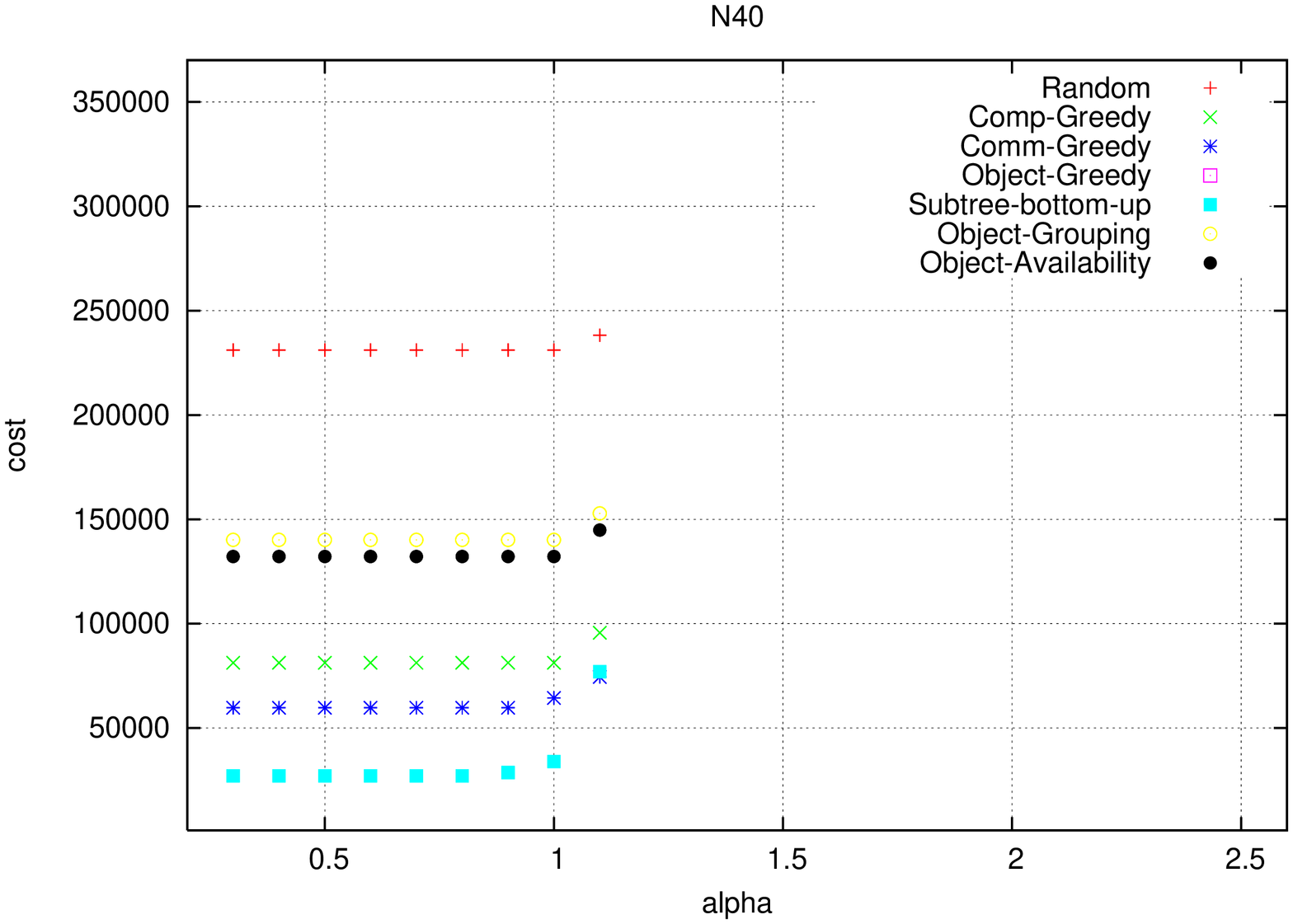}
     \label{fig.exp6-40}
   }
   \caption{Simulation with big basic objects and high download
   rates, increasing $\alpha$.}
\label{fig.exp6}
\end{figure}

\paragraph{Low download rates - small object sizes}
The behavior of the heuristics when download rates are low, i.e.,
frequency = 1/50s, is almost the same as for high download rates. In
general the heuristics lead to the same operator mapping, but in some cases
the purchased processors have less powerful network cards
(Cf.~Table~\ref{tab.exp3}).

\begin{center}
\begin{table}[Ht]
\caption{Influence of the download rate on the platform cost, in \$, when object sizes are small.}
\label{tab.exp3}
\begin{center}
\begin{tabular}{|c|c|c|c||c|c|c|}
\hline
& \multicolumn{3}{|c||}{small object sizes} & \multicolumn{3}{|c|}{big
object sizes}\\
\hline
N  & Comm-Greedy & Obj-Greedy & Subtree-b-up & Comm-Greedy &
Obj-Greedy & Subtree-b-up\\
\hline
\hline
115    &   \textbf{7947}  &  13547 &   8745  &   \textbf{7548}   & 13547   & 8745  \\
116    &  \textbf{15495}  &  13547 &   \textbf{7947}  &  \textbf{15096}   & 13547   & \textbf{7548}  \\
117    &   \textbf{7947}  &  13547 &   \textbf{7947}  &   \textbf{7548}   & 13547   & \textbf{7548}  \\
118    &  \textbf{15495}  &  13547 &   7548  &  \textbf{15096}   & 13547   & 7548  \\
119    &  \textbf{15495}  &  13547 &   8745  &  \textbf{15096}   & 13547   & 8745  \\

\hline
\end{tabular}
\end{center}
\end{table}
\end{center}


\paragraph{Low download rates - big object sizes}
In this case, low download rates slightly improve the success rate of the
heuristics (see Figure~\ref{fig.exp7}). Indeed, because of the lower
download rates the links between servers and processors are less congested,
and hence the server association is feasible in more scenarios.

\begin{figure}
   \centering
   \subfigure[$frequency = 1/50s$.]{
     \includegraphics[angle=270, width=0.45\textwidth]{exp5-0-9.eps}
     \label{fig.exp5-0-9}
   }$\quad$
   \subfigure[$frequency = 1/2s$.]{
     \includegraphics[angle=270, width=0.45\textwidth]{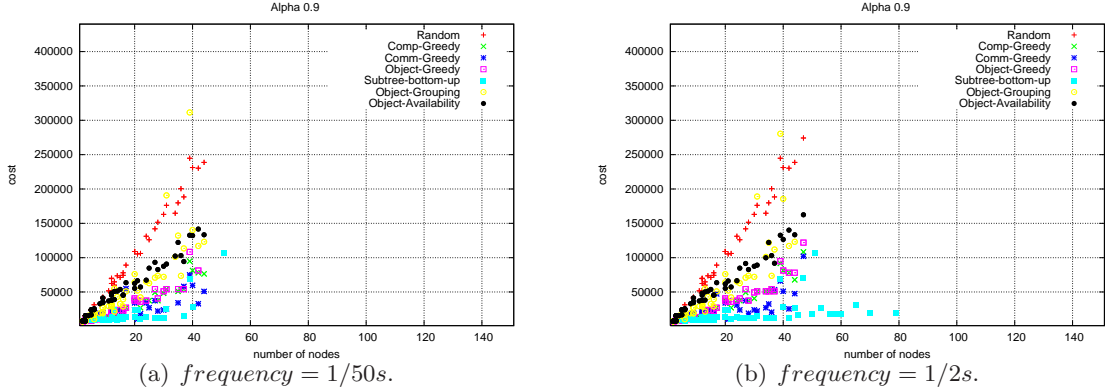}
     \label{fig.exp7-0.9}
   }
   \caption{Comparison between the simulations with big basic
   objects and different download rates, $\alpha = 0.9$.}
\label{fig.exp7}
\end{figure}

\paragraph{Influence of download rates (frequency) on the solution}
The third set of experiments studies the influence of download rates
on the solution. Remember the download rate of a basic object $k$ is
computed by $rate_k = frequency \times \delta_k$.
A first results is that frequencies smaller than $1/10s$ has no
further influence on the solution. All heuristics will find the same
solutions for a fixed operator tree, as seen in Figure~\ref{fig.exp9}.
For frequencies between $1/2s$ and $1/10s$, the solution
cost changes. In general the cost decreases, but for $N=160$ the cost
for the Object-Grouping heuristic increases.
Furthermore, the heuristic ranking remains:
Subtree-bottom-up, followed by the Greedy family, followed by the
object sensitive ones, and Random forms the bottom of the
league. Interestingly, the costs of Object-Availability decrease with
the number of operators. In this case the number of operators
that need to download a basic object increases, and hence the
privileged treatment of basic objects in order of availability on
servers becomes more important (compare
Figure~\ref{fig.exp9} and Figure~\ref{fig.exp9-100}).
\begin{figure}
   \centering
   \subfigure[$N = 140$.]{
    \includegraphics[angle=270, width=0.45\textwidth]{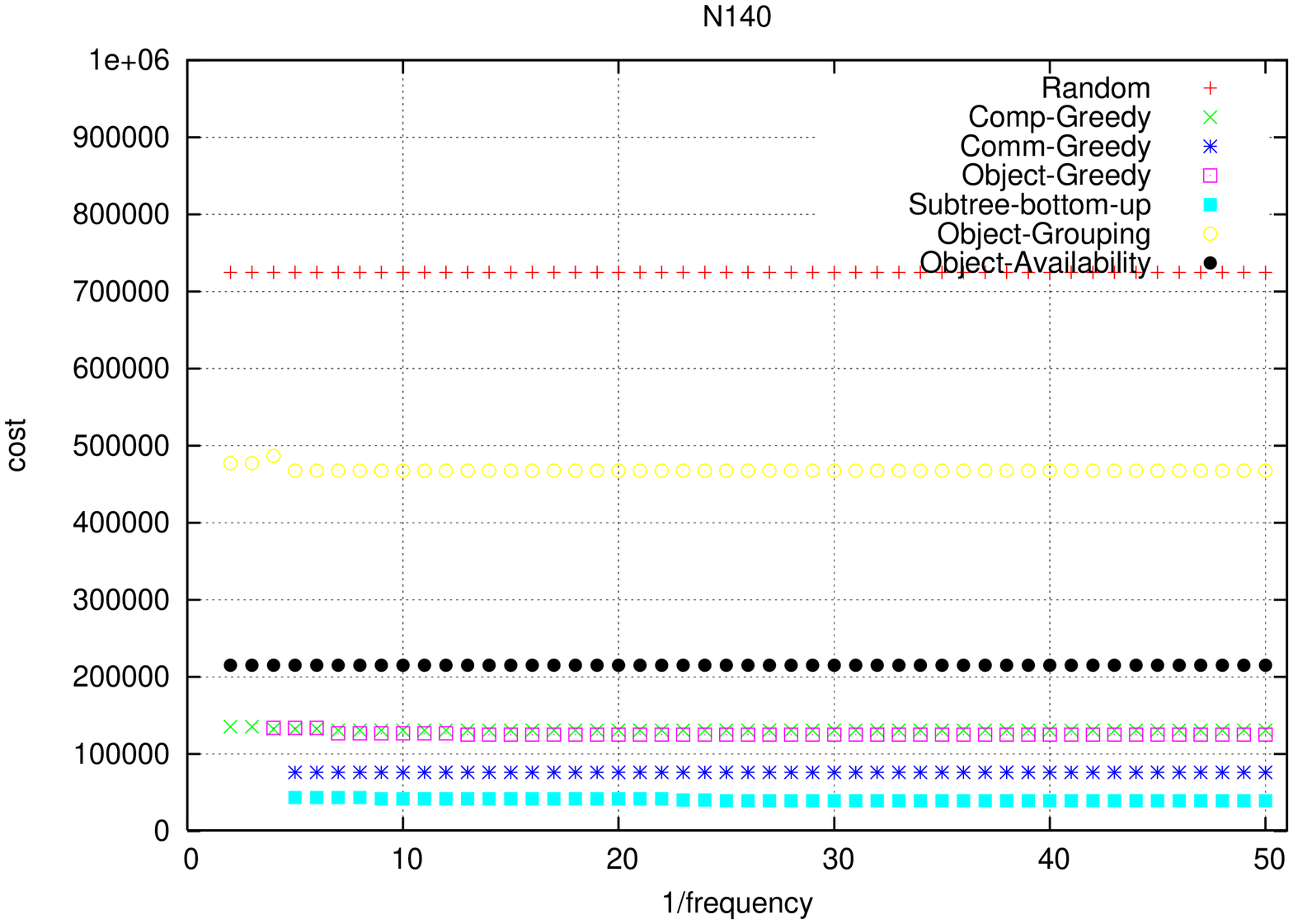}
     \label{fig.exp9-140}
   }$\quad$
   \subfigure[$N = 160$.]{
     \includegraphics[angle=270, width=0.45\textwidth]{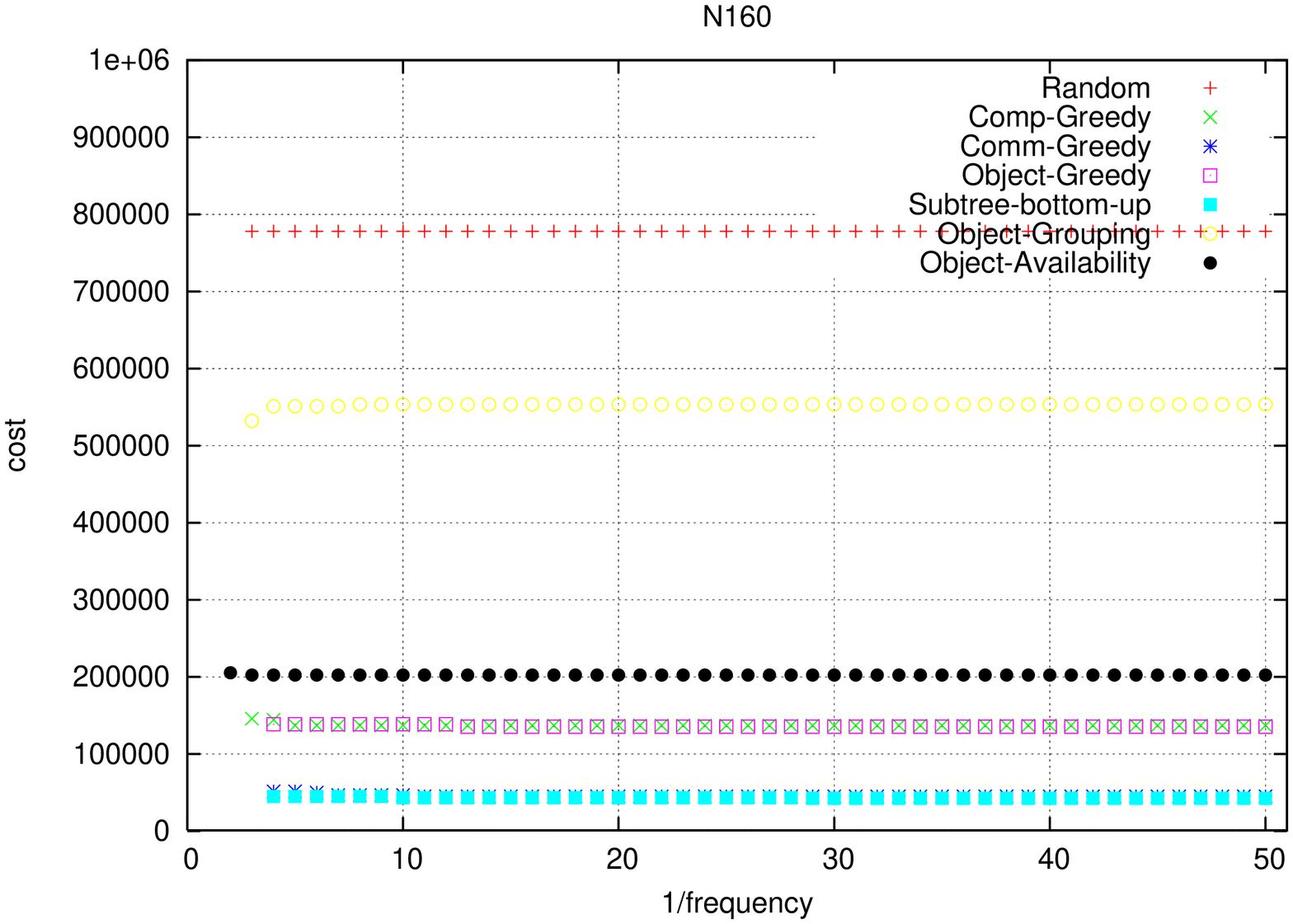}
     \label{fig.exp9-160}
   }
   \caption{Simulation to evaluate the influence of the frequency that determines
   download rates.}
\label{fig.exp9}
\end{figure}

We also tested the importance of the number of basic object
replications on the servers. Initially we ran
experiments also on different server configurations, with basic
objects either not replicated are replicated on all servers.
However, we did not observe a significant difference in the results
across different server configurations. We thus present results
only for are default server configuration.
Figure~\ref{fig.exp9} shows results for decreasing
frequencies, when each basic object is available only on a single
server. Comparing this plot to Figure~\ref{fig.exp9-50-100}, for
which each basic object is available on 50\% of the servers, one notices
no significant difference. Focusing solely on
frequencies between $1/2s$ and $1/10s$, we see that
Subtree-bottom-up, Comm-Greedy, Object-Grouping, and Object-Greedy find
more solutions, at frequencies for which they failed before (Figure~\ref{fig.exp9-100}). We
conclude that the level of replication of basic objects on servers may
matter for application trees with specific structures and
download frequencies, but that in general we can consider that this
parameter has little or no effect on the performance of the heuristics.
\begin{figure}
   \centering
   \subfigure[Each basic object available only on one server.]{
    \includegraphics[angle=270, width=0.45\textwidth]{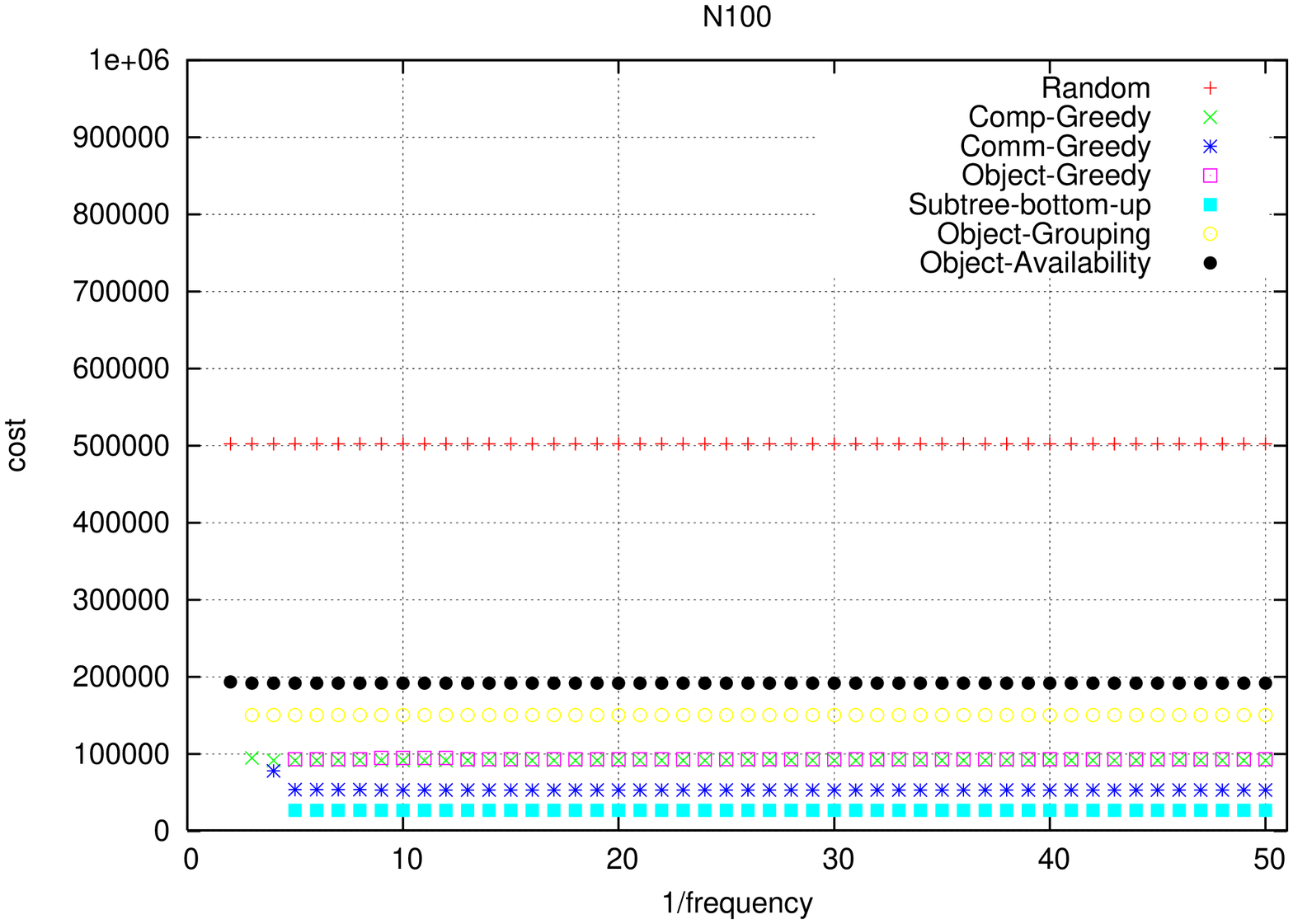}
     \label{fig.exp9-100}
   }$\quad$
   \subfigure[Each basic object available on 50\% of the servers.]{
     \includegraphics[angle=270, width=0.45\textwidth]{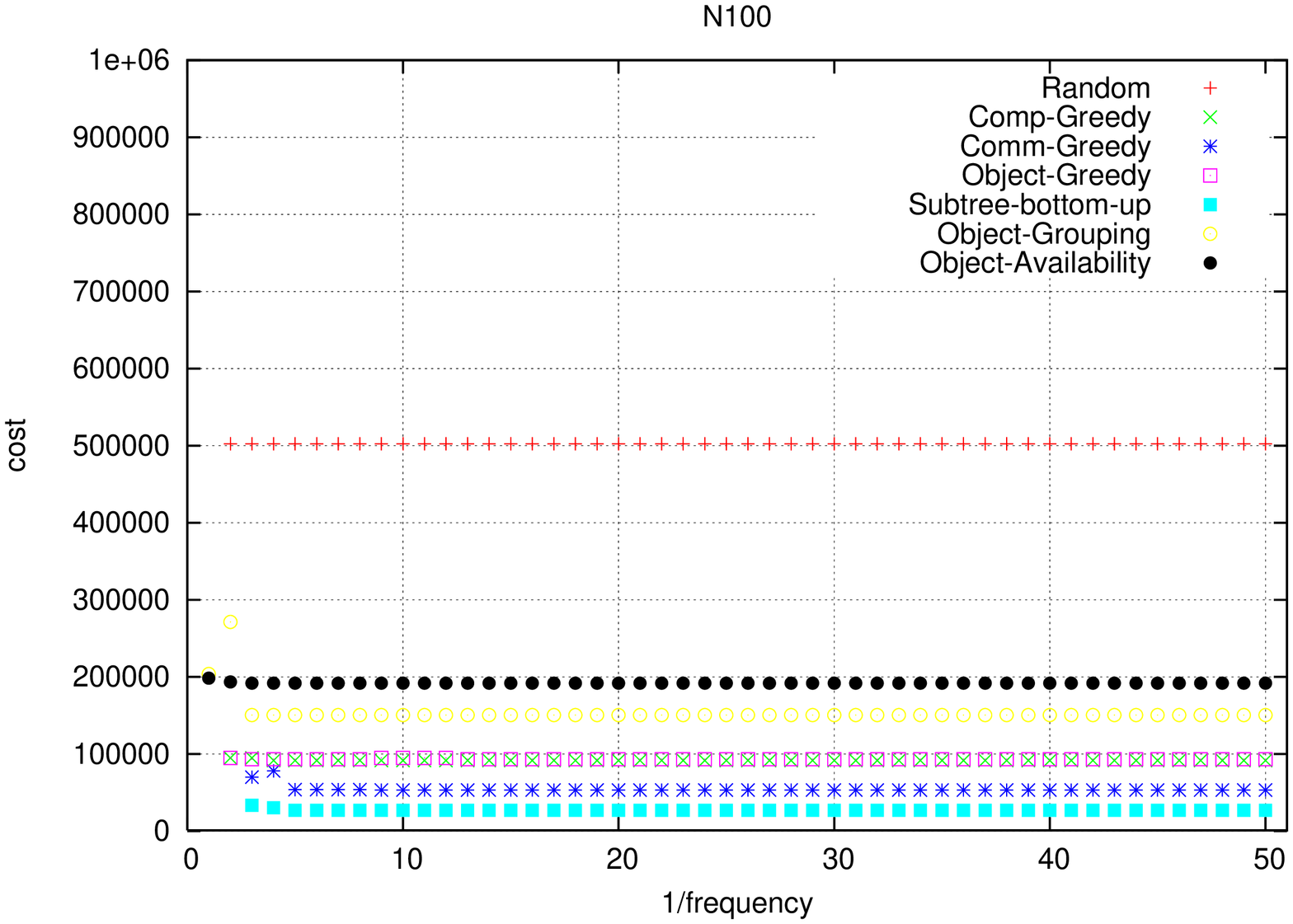}
     \label{fig.exp9-50-100}
   }
   \caption{Simulation to evaluate the influence of the availability of
   the basic objects.}
\label{fig.exp9obj}
\end{figure}

\paragraph{Comparison of the heuristics to a LP solution on a
   homogeneous platform}

This last set of experiments is dedicated to the evaluation of our
heuristics via a lower bound given by the solution of our integer
linear program.
We use Cplex 11 to solve our linear program. Unfortunately, the
LP is so enormous that, even when using only 5 possible groups of
processors and using tress with 30 operators, the LP file could not be
opened in Cplex. For trees with 20 operators, Cplex produces the optimal
solution, which consists in all cases in buying a single
processor. So we opted for evaluating our heuristics vs. the optimal
solution under homogeneous conditions, i.e., when there is
a single processor type. In this case we
skip the downgrade step after the server allocation.
When $\alpha$ is less than 1, Subtree-bottom-up almost always finds the
optimal solution (see Figure~\ref{fig.exp11-0.9}). Note that once
again, in two cases this heuristic is not able to find a feasible solution,
while the others succeed ($N\in\{34,35,36\}$). This is again due to the
fact that Subtree-bottom-up maps all operators onto a single processor and
then the server association process fails. The other heuristics buy
more processors from the onset, and are later able to find a feasible
processor-server association.

Even with homogeneous conditions, we observe the same ranking of our
heuristics as before: Subtree-bottom up, the Greedy family, followed
by Object-Grouping, then Object-Availability and finally
Random. Focusing on the Greedy family, we observe that with increasing
operator trees, Comp-Greedy outperforms Object-Greedy, and in most
cases Comm-Greedy achieves the best costs of the three.

\begin{figure}
   \centering
   \subfigure[$\alpha = 0.9$.]{
    \includegraphics[angle=270, width=0.45\textwidth]{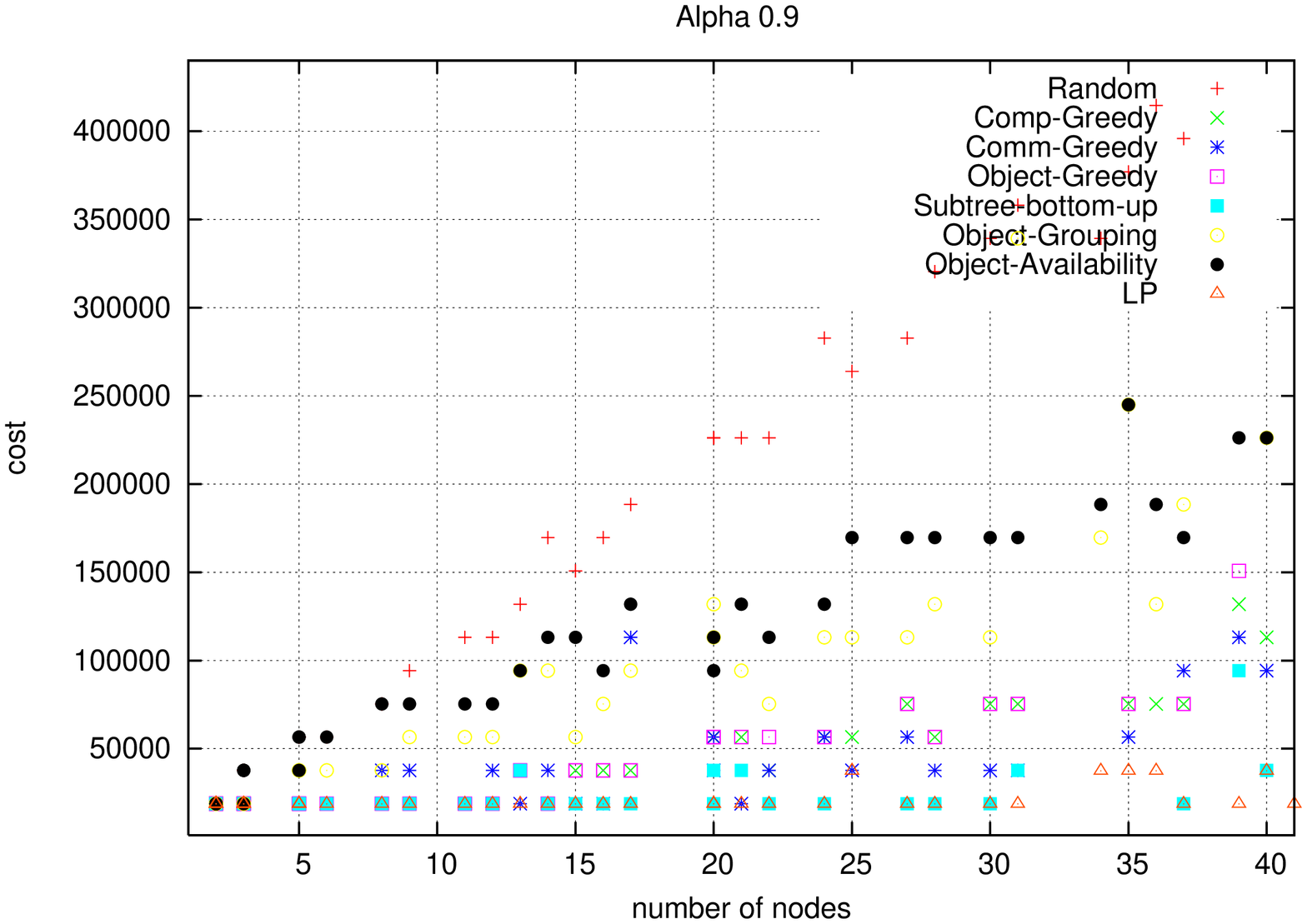}
     \label{fig.exp11-0.9}
   }$\quad$
   \subfigure[$\alpha = 1.1$.]{
     \includegraphics[angle=270, width=0.45\textwidth]{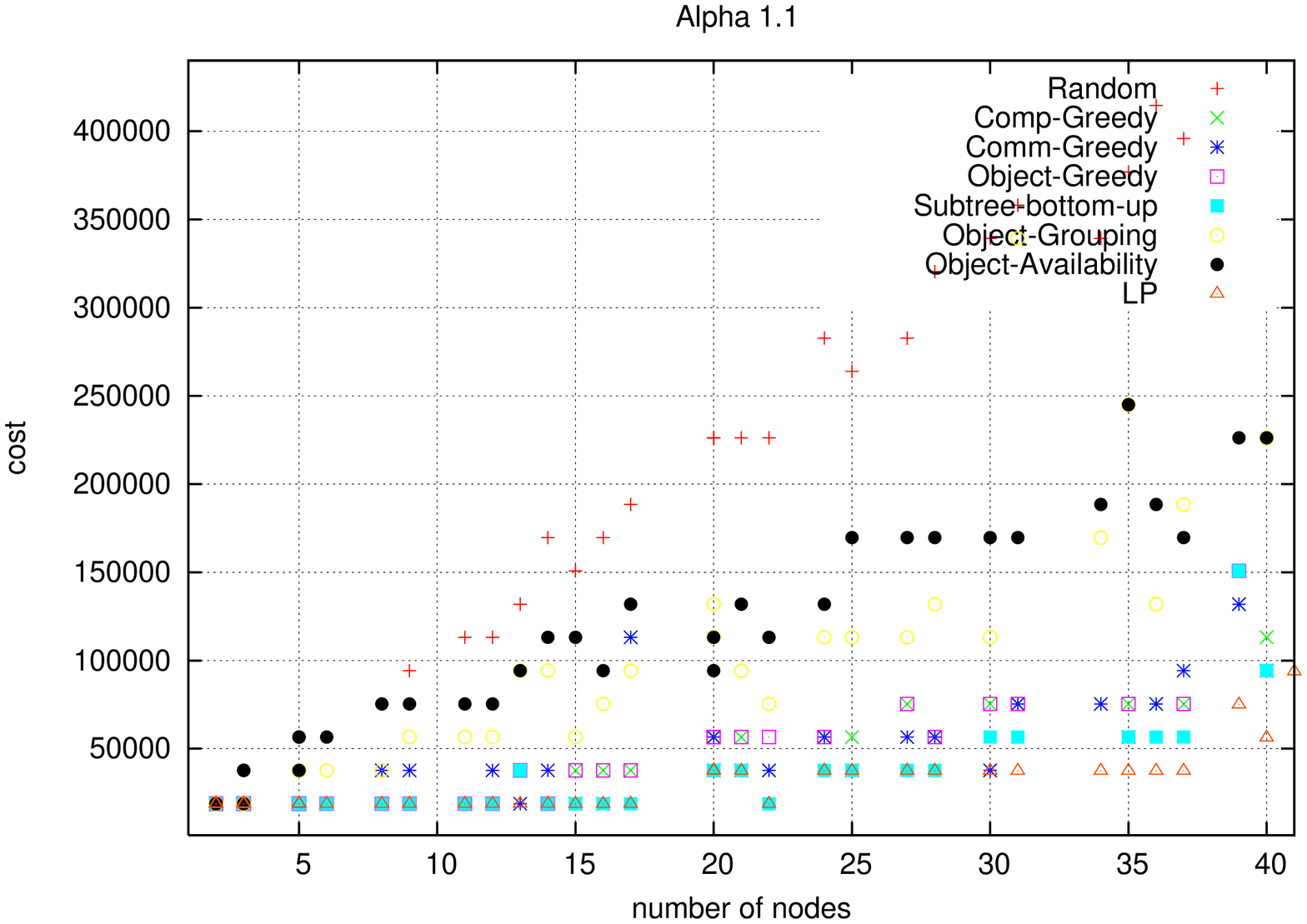}
     \label{fig.exp11-1.1}
   }
   \caption{Simulation to compare the heuristics' performances to the
   LP performance on homogeneous platforms.}
\label{fig.exp11}
\end{figure}

\paragraph{Summary}
Our results show that all our more sophisticated heuristics perform better
than the simple random approach.  Unfortunately, the object sensitive
heuristics, Object-Grouping and Object-Availability, do not show the
desired performance. We think that in some situations these heuristics
could lead to good performance, but this is not observed on our set of
random application configurations.  We had found that Subtree-bottom-up
outperforms other heuristics in most situations and also produces results
very close to the optimal (for the cases in which we were able to determine
the optimal).  There are some configurations for which Subtree-bottom-up
fails, our results  suggest that on should use one of our Greedy
heuristics, which perform reasonably well.

\section{Conclusion}
\label{sec.conclusion}

In this paper we have studied the problem of resource allocation for
in-network stream processing. We formalized several operator-placement
problems.  We have focused more particularly on a ``constructive'' scenario
in which one aims at minimizing the cost of a platform that satisfies an
application throughput requirement.  The complexity analysis showed that all
problems are NP-complete, even for the simpler cases.  We have derived an
integer linear programming formulation of the various problems, and we have
proposed several polynomial time heuristics for the constructive scenario.
We compared these heuristics through simulation, allowing us to identify
one heuristic that is almost always better than the others,
Subtree-bottom-up.  Finally, we assessed the absolute performance of our
heuristics with respect to the optimal solution of the linear program for
homogeneous platforms and small problem instances. It turns out that the
Subtree-bottom-up heuristic almost always produces optimal results.

An interesting direction for future work is the study of the case when
multiple applications must be executed simultaneously so that a given
throughput must be achieved for each application. In this case a clear
opportunity for higher performance with a reduced cost is the reuse of
common sub-expression between trees~\cite{Pandit06, Munagala_PODS2007}.
Another direction is the study of applications that are mutable, i.e.,
whose operators can be rearranged based on operator associativity and
commutativity rules. Such situations arise for instance in relational
database applications~\cite{chen02design}.

\bibliographystyle{IEEEtran}
\bibliography{biblio}

\begin{thebibliography}{10}
\providecommand{\url}[1]{#1}
\csname url@rmstyle\endcsname
\providecommand{\newblock}{\relax}
\providecommand{\bibinfo}[2]{#2}
\providecommand\BIBentrySTDinterwordspacing{\spaceskip=0pt\relax}
\providecommand\BIBentryALTinterwordstretchfactor{4}
\providecommand\BIBentryALTinterwordspacing{\spaceskip=\fontdimen2\font plus
\BIBentryALTinterwordstretchfactor\fontdimen3\font minus
  \fontdimen4\font\relax}
\providecommand\BIBforeignlanguage[2]{{%
\expandafter\ifx\csname l@#1\endcsname\relax
\typeout{** WARNING: IEEEtran.bst: No hyphenation pattern has been}%
\typeout{** loaded for the language `#1'. Using the pattern for}%
\typeout{** the default language instead.}%
\else
\language=\csname l@#1\endcsname
\fi
#2}}

\bibitem{badcock_VLDB_2004}
B.~Badcock, S.~Babu, M.~Datar, R.~Motwani, and J.~Widom, ``{Models and issues
  in data stream systems},'' in \emph{Proceedings of the Intl. Conf. on Very
  Large Data Bases}, 2004, pp. 456--467.

\bibitem{srivastava_PODS2005}
U.~Srivastava, K.~Munagala, and J.~Widom, ``{Operator Placement for In-Network
  Stream Query Processing},'' in \emph{Proceedings of the 24th Intl. Conf. on
  Principles of Database Systems}, 2005, pp. 250--258.

\bibitem{madden_ICMD_2003}
S.~Madden, M.~Franklin, J.~Hellerstein, and W.~Hong, ``{The design of an
  acquisitional query processor for sensor networks},'' in \emph{Proceedings of
  the 2003 ACM SIGMOD Intl. Conf. on Management of Data}, 2003, pp. 491--502.

\bibitem{bonnet_CMDB_2001}
P.~Bonnet, J.~Gehrke, and P.~Seshadri, ``{Towards sensor database systems},''
  in \emph{Proceedings of the Conference on Mobile Data Management}, 2001.

\bibitem{cranor_ICMD_2002}
C.~Cranor, Y.~Gao, T.~Johnson, V.~Shkapenyuk, and O.~Spatscheck, ``{Gigascope:
  high-performance network monitoring with an SQL interface},'' in
  \emph{Proceedings of the ACM SIGMOD International Conference on Management of
  Data}, 2002, pp. 623--633.

\bibitem{vanRennesse_IPTPS_2002}
R.~van Rennesse, K.~Birman, D.~Dumitriu, and W.~Vogels, ``{Scalable Management
  and Data Mining Using Astrolabe},'' in \emph{Proceedings from the First
  International Workshop on Peer-to-Peer Systems}, 2002, pp. 280--294.

\bibitem{cooke_USENIX_2006}
E.~Cooke, R.~Mortier, A.~Donnelly, P.~Barham, and R.~Isaacs, ``{Reclaiming
  Network-wide Visibility Using Ubiquitous End System Monitors},'' in
  \emph{Proceedings of the USENIX Annual Technical Conference}, 2006.

\bibitem{Babu_SIGMODRECORD_2001}
S.~Babu and J.~Widom, ``{Continuous Queries over Data Streams},'' \emph{SIGMOD
  Record}, vol.~30, no.~3, 2001.

\bibitem{liu_tkde1999}
L.~Liu, C.~Pu, and W.~Tang, ``{Continual Queries for Internet Scale
  Event-Driven Information Delivery},'' \emph{IEEE Transactions on Knowledge
  and Data Engineering}, vol.~11, no.~4, pp. 610--628, 1999.

\bibitem{chen02design}
J.~Chen, D.~J. DeWitt, and J.~F. Naughton, ``{Design and Evaluation of
  Alternative Selection Placement Strategies in Optimizing Continuous
  Queries},'' in \emph{{Proceedings of ICDE}}, 2002.

\bibitem{Plale_TPDD_2003}
B.~Plale and K.~Schwan, ``{Dynamic Querying of Streaming Data with the dQUOB
  System},'' \emph{IEEE Transactions on Parallel and Distributed Systems},
  vol.~14, no.~4, pp. 422--432, 2003.

\bibitem{kramer_COMAD05}
J.~Kr\"ame and B.~Seeger, ``{A Temporal Foundation for Continuous Queries over
  Data streams},'' in \emph{Proceedings of the Intl. Conf. on Management of
  Data}, 2005, pp. 70--82.

\bibitem{abadi2005design}
D.~J. Abadi, Y.~Ahmad, M.~Balazinska, U.~Cetintemel, M.~Cherniack, J.-H. Hwang,
  W.~Lindner, A.~S. Maskey, A.~Rasin, E.~Ryvkina, N.~Tatbul, Y.~Xing, and
  S.~Zdonik, ``{The Design of the Borealis Stream Processing Engine},'' in
  \emph{Second Biennial Conference on Innovative Data Systems Research (CIDR
  2005)}, Asilomar, CA, January 2005.

\bibitem{medusa}
M.~Cherniack, H.~Balakrishnan, M.~Balazinska, D.~Carney, U.~Cetintemel,
  Y.~Xing, and S.~Zdonik, ``Scalable distributed stream processing,'' in
  \emph{Proc. of the CIDR Conf.}, January 2003.

\bibitem{pier}
\BIBentryALTinterwordspacing
R.~Huebsch, J.~M. Hellerstein, N.~L. Boon, T.~Loo, S.~Shenker, and I.~Stoica,
  ``{Querying the Internet with PIER},'' Sept. 2003. [Online]. Available:
  \url{citeseer.ist.psu.edu/huebsch03querying.html}
\BIBentrySTDinterwordspacing

\bibitem{gates}
L.~Chen, K.~Reddy, and G.~Agrawal, ``{GATES: a grid-based middleware for
  processing distributed data streams},'' \emph{High performance Distributed
  Computing, 2004. Proceedings. 13th IEEE International Symposium on}, pp.
  192--201, 4-6 June 2004.

\bibitem{nath-irisnet}
\BIBentryALTinterwordspacing
S.~Nath, A.~Deshpande, Y.~Ke, P.~B. Gibbons, B.~Karp, and S.~Seshan,
  ``{IrisNet: An Architecture for Internet-scale Sensing Services}.'' [Online].
  Available: \url{citeseer.ist.psu.edu/nath03irisnet.html}
\BIBentrySTDinterwordspacing

\bibitem{NiagaraCQ}
J.~Chen, D.~DeWitt, F.~Tian, and Y.~Wang, ``{NiagaraCQ: A scalable continuous
  query system for internet databases},'' in \emph{Proceedings of the SIGMOD
  Intl. Conf. on Management of Data}, 2000, pp. 379--390.

\bibitem{MORTAR}
D.~Logothetis and K.~Yocum, ``{Wide-Scale Data Stream Management},'' in
  \emph{Proceedings of the USENIX Annual Technical Conference}, 2008.

\bibitem{Pietzuch_ICDE06}
P.~Pietzuch, J.~Leflie, J.~Shneidman, M.~Roussopoulos, M.~Welsh, and
  M.~Seltzer, ``{Network-Aware Operator Placement for Stream-Processing
  Systems},'' in \emph{Proceedings of the 22nd International Conference on Data
  Engineering (ICDE'06)}, 2006, pp. 49--60.

\bibitem{ahmad_VLDB_2004}
Y.~Ahmad and U.~Cetintemel, ``Network aware query processing for stream-based
  applications,'' in \emph{Proceedings of the International Conference on Very
  Large Data Bases}, 2004, pp. 456--467.

\bibitem{ioannidis96query}
Y.~E. Ioannidis, ``Query optimization,'' \emph{ACM Computing Surveys}, vol.~28,
  no.~1, pp. 121--123, 1996.

\bibitem{chaudhuri98overview}
S.~Chaudhuri, ``{An Overview of Query Optimization in Relational Systems},'' in
  \emph{Proc. 17th {ACM} symposium on Principles of Database Systems}, 1998,
  pp. 34--43.

\bibitem{Deshpande2007}
A.~Deshpande, Z.~G. Ives, and V.~Raman, ``Adaptive query processing,''
  \emph{Foundations and Trends in Databases}, vol.~1, no.~1, pp. 1--140, 2007.

\bibitem{HongPrasanna07}
B.~Hong and V.~K. Prasanna, ``Adaptive allocation of independent tasks to
  maximize throughput,'' \emph{IEEE Trans. Parallel Distributed Systems},
  vol.~18, no.~10, pp. 1420--1435, 2007.

\bibitem{GareyJohnson}
M.~R. Garey and D.~S. Johnson, \emph{Computers and Intractability, a Guide to
  the Theory of {NP}-Completeness}.\hskip 1em plus 0.5em minus 0.4em\relax W.H.
  Freeman and Company, 1979.

\bibitem{codeVeroQuery}
\BIBentryALTinterwordspacing
``{Source Code for the Heuristics}.'' [Online]. Available:
  \url{http://graal.ens-lyon.fr/~vsonigo/code/query-streaming/}
\BIBentrySTDinterwordspacing

\bibitem{Pandit06}
V.~Pandit and H.~Ji, ``Efficient in-network evaluation of multiple queries,''
  in \emph{HiPC}, 2006, pp. 205--216.

\bibitem{Munagala_PODS2007}
K.~Munagala, U.~Srivastava, and J.~Widom, ``Optimization of continuous queries
  with shared expensive filters,'' in \emph{PODS '07: Proceedings of the
  twenty-sixth ACM SIGMOD-SIGACT-SIGART symposium on Principles of database
  systems}.\hskip 1em plus 0.5em minus 0.4em\relax New York, NY, USA: ACM,
  2007, pp. 215--224.

\end{thebibliography}
\end{document}